\documentclass[12pt]{article}
\usepackage{epsfig,amsfonts,amssymb}
\usepackage{mathtools}  

\usepackage{hyperref}
\usepackage{tabulary}
\usepackage{booktabs}

\usepackage{amsthm}
 
\theoremstyle{definition}
\newtheorem{claim}{Claim}[section]
 
\theoremstyle{remark}

\usepackage{amsmath}
\usepackage{comment}
\usepackage{comment}
\usepackage{cite}
\usepackage{comment}
\usepackage{cite}
\usepackage{float}
\usepackage{graphicx}
\newcommand{\hook}{\reflectbox{\reflectbox{\checkmark}}}
\newcommand{\unhook}{\reflectbox{\checkmark}}
\newcommand{\cohook}{\rotatebox[origin=c]{180}{\hook}}
\newcommand{\counhook}{\rotatebox[origin=c]{180}{\unhook}}

\def\ZZZ{{\hbox{ Z\kern-1.6mm Z}}}
\def\RRR{{\hbox{ R\kern-2.4mm R}}}
\def\CCC{{\hbox{ C\kern-2.0mm C}}}
\def\zzz{{\hbox{z\kern-1mm z}}}

\newcommand{\qeq}{{\hbox{=\kern-2.3mm ? \kern.5mm }}}
\renewcommand{\qeq}{=}

\newcommand{\be}{\begin{equation}}
\newcommand{\ee}{\end{equation}}
\newcommand{\ben}{\begin{equationarray}\displaystyle}
\newcommand{\een}{\end{equationarray}}

\def\one{{\hbox{ 1\kern-.8mm l}}}
\def\zero{{\hbox{ 0\kern-1.5mm 0}}}

\newcommand{\bea}[1]{\begin{equationarray}\label{#1} }
\newcommand{\eea}{\end{equationarray}}

\usepackage{bm}
\usepackage[table]{xcolor}

\def\figone{

\def\JPicScale{0.4}
\ifx\JPicScale\undefined\def\JPicScale{1}\fi
\unitlength \JPicScale mm
\begin{picture}(140,90)(0,0)
\linethickness{0.1mm}
\put(30,60){\line(1,0){100}}
\linethickness{0.1mm}
\put(30,20){\line(1,0){100}}
\linethickness{0.3mm}
\multiput(30,90)(0.12,-0.18){167}{\line(0,-1){0.18}}
\linethickness{0.3mm}
\multiput(100,60)(0.12,0.12){250}{\line(1,0){0.12}}
\linethickness{0.3mm}
\multiput(100,20)(0.18,-0.12){167}{\line(1,0){0.18}}
\linethickness{0.3mm}
\multiput(30,0)(0.18,0.12){167}{\line(1,0){0.18}}
\put(150,62){\makebox(0,0)[cc]{$x^{\prime0}=T$}}

\put(155,22){\makebox(0,0)[cc]{$x^{\prime0}=-T$}}

\put(95,65){\makebox(0,0)[cc]{$c_{(1)}$}}

\put(140,90){\makebox(0,0)[cc]{$r_{(1)}$}}

\put(20,90){\makebox(0,0)[cc]{$r_{(2)}$}}

\put(140,0){\makebox(0,0)[cc]{$r_{(3)}$}}

\put(20,0){\makebox(0,0)[cc]{$r_{(4)}$}}

\put(55,65){\makebox(0,0)[cc]{$c_{(2)}$}}

\put(95,15){\makebox(0,0)[cc]{$c_{(3)}$}}

\put(65,15){\makebox(0,0)[cc]{$c_{(4)}$}}

\end{picture}

}

\begin{document}

\baselineskip 24pt

\begin{center}

{\Large \bf On Positive Geometries of Quartic Interactions II :\\ Stokes polytopes, Lower Forms on Associahedra and World-sheet Forms}

\vspace*{2.0ex} 

\centerline{\fontsize{13pt}{13pt}\selectfont 
 \rm P B Aneesh$^a$, Pinaki Banerjee$^b$, Mrunmay Jagadale$^{a,c}$, Renjan Rajan John$^{d,e}$, Alok Laddha$^a$, Sujoy Mahato$^f$}


\end{center}

\fontsize{12pt}{12pt}\selectfont



\baselineskip=18pt

\vspace*{2.0ex}

\centerline{\small \it ~$^a$ Chennai Mathematical Institute, H1, SIPCOT IT Park, Siruseri, Kelambakkam 603103, India}
\centerline{\small \it ~$^b$ Indian Institute of Technology Kanpur, Kalyanpur, Kanpur 208016, India}
\centerline{\small \it~$^c$ Tata Institute of Fundamental Research, Homi Bhabha Road, 
Mumbai 400005, India}
\centerline{\small \it~$^d$ Universit\`a del Piemonte Orientale, Dipartimento di Scienze e Innovazione Tecnologica}
\centerline{\small \it Viale T. Michel 11, I-15121 Alessandria, Italy }
\centerline{\small \it~$^e$  I.\,N.\,F.\,N. - sezione di Torino, Via P. Giuria 1, I-10125 Torino, Italy}
\centerline{\small \it~$^f$ Institute of Mathematical Sciences, Homi Bhabha National Institute}
\centerline{\small \it IV Cross Road, C. I. T. Campus, Taramani, Chennai 600113, India}

\vspace*{2.0ex} 

\centerline{\small{aneeshpb@cmi.ac.in, pinakib@iitk.ac.in, jmrunmay@cmi.ac.in,}}
\centerline{\small{renjan.rajan@to.infn.it, aladdha@cmi.ac.in, sujoymahato@imsc.res.in}}

\vspace*{2.0ex}

\fontsize{12pt}{17pt}\selectfont

\centerline{\bf Abstract} \bigskip

In \cite{Banerjee:2018tun}, two of the present authors along with P. Raman attempted to extend the Amplituhedron program for scalar field theories \cite{Arkani-Hamed:2017mur} to quartic scalar interactions. In this paper we develop various aspects of this proposal.  Using recent seminal results in Representation theory \cite{1807ppp,1906ppp}, we show that projectivity of scattering forms and existence of kinematic space associahedron completely capture planar amplitudes of quartic interaction. We generalise the results of \cite{Banerjee:2018tun} and show that for any $n$-particle amplitude, the positive geometry associated to the projective scattering form is a convex realisation of Stokes polytope which can be naturally embedded inside one of the ABHY associahedra defined in \cite{Arkani-Hamed:2017mur,HughThomas}. For a special class of Stokes polytopes with hyper-cubic topology, we show that they have a canonical convex realisation in kinematic space  as boundaries of kinematic space associahedra.

We then use these kinematic space geometric constructions to write world-sheet forms for $\phi^{4}$ theory which are forms of lower rank on the CHY moduli space. We argue that just as in the case of bi-adjoint $\phi^3$ scalar amplitudes, scattering equations can be used as diffeomorphisms between certain $\frac{n-4}{2}$ forms on the world-sheet and $\frac{n-4}{2}$ forms on ABHY associahedron that generate quartic amplitudes.

\vfill \eject

\fontsize{12pt}{18pt}\selectfont

\tableofcontents

\section{Introduction}

\cite{Arkani-Hamed:2017mur} brought the Amplituhedron program and the world of polytopes in contact with scattering amplitudes for non-supersymmetric Quantum field theories. In a number of beautiful results established in \cite{Arkani-Hamed:2017mur}, it was shown that the amplituhedron for bi-adjoint scalar $\phi^{3}$ theory is a particular realisation of a combinatorial polytope called associahedron ${\cal A}_{n}$ which lives in the kinematic space of Mandelstam invariants. It was also shown that the unique canonical top form on the associahedron gives the scattering amplitude of the theory. This new understanding of scattering amplitude which paralleled striking developments in the supersymmetric world \cite{Arkani-Hamed:2013jha,Arkani-Hamed:2013kca,Franco:2014csa,Arkani-Hamed:2018rsk,Arkani-Hamed:2017vfh}, also gave rise to a novel understanding/derivation of the CHY formula for bi-adjoint scalar field theories \cite{Arkani-Hamed:2017mur, He:2018pue}. 

The CHY formulae are written as integrals on the moduli space of punctured Riemann sphere, ${\cal M}_{0,n}$. This space admits a real section, ${\cal M}_{0,n}({\bf R})$ which is parametrized by an (equivalence class of) ordered set of points on a disk \footnote{This space is also known as positive moduli space and often denoted as ${\cal M}_{0,n}^{+}$ in the literature.} . It was shown in\cite{Arkani-Hamed:2017mur, delaCruz:2017zqr} that there exists a compactification $\overline{{\cal M}}_{0,n}({\bf R})$  of ${\cal M}_{0,n}({\bf R})$ which is also an associahedron.  CHY scattering equations precisely generate a diffeomorphism between the world-sheet associahedron $\overline{{\cal M}}_{0,n}({\bf R})$ and the kinematic space associahedron ${\cal A}_{n}$ such that the CHY formula for scattering amplitude can be understood in terms of pushforward of the canonical form between the two associahedra. 

In kinematic space, attempts were made to extend the amplituhedron program to generic scalar interactions \cite{Banerjee:2018tun,Raman:2019utu,Jagadale:2019byr}. In \cite{Banerjee:2018tun} it was shown that planar tree-level amplitudes for massless quartic interactions in any space-time dimensions were intricately tied to another polytope called Stokes polytope. Just as in the case of cubic interactions and the associahedron in kinematic space, it was argued that any $n$ particle amplitude can be obtained from canonical forms on $\frac{n-4}{2}$ dimensional Stokes polytopes.  Each Stokes polytope is parametrized by a reference quadrangulation $Q$ of an $n$-gon such that the corresponding canonical form produces a ``partial" scattering amplitude $m_{n}^{Q}$. It was argued in \cite{Banerjee:2018tun} that a (weighted) sum over $m_{n}^{Q}$ produces the full planar tree-level amplitude. These claims were verified for $n=6,8,\textrm{and}\ 10$ particle scattering. 

In the world of CHY formalism, it was shown in a series of seminal works \cite{Cachazo:2014xea,Baadsgaard:2015ifa,Baadsgaard:2016fel} that integrands for any monomial scalar interaction could be written down \footnote{This came as a surprise as CHY formalism is intimately tied to perturbative string amplitudes and due to the absence of any  independent quartic vertex in perturbative string amplitude, it was perhaps expected that no formulae could be written down in CHY formalism which were intrinsically tied to generic scalar field theories. We thank Poul Damgaard and Ashoke Sen for discussion on this issue.} . These results fit neatly into the basic paradigm of CHY formula for scattering amplitudes where the underlying moduli space (moduli space of punctured Riemann sphere) is universal but the integrand (which is an $n-3$ form for an $n$ particle scattering) depends on the theory. Ideologically, this appears to be in contrast with the emerging picture in kinematic space where for every scalar interaction the corresponding polytope is a unique member of the ``accordiohedron family" \cite{accoref} and the form that corresponds to the amplitude is the unique canonical form associated to the accordiohedron. 

This contrast raises a few questions. Just as the derivation of amplitude in bi-adjoint $\phi^3$ scalar theory from canonical form on the associahedron was a way to re-formulate the CHY formula, {\bf (a)} can one interpret kinematic space Stokes polytopes and their canonical forms from the moduli space perspective, and {\bf (b)} can one write down world-sheet formulae for massless planar $\phi^{4}$ amplitudes which are ``derived" from the corresponding canonical forms in kinematic space.  There appears a major obstacle in quantifying these questions :
\begin{itemize}
\item Derivation of the CHY formula for bi-adjoint $\phi^3$ scalar amplitude as a pushforward map between associahedra relied on the key result that compactification of the real moduli space ${\cal M}_{0,n}({\bf R})$ is also an associahedron and that the corresponding form on it is precisely the well known Parke-Taylor form. However, no such moduli space is known for which the corresponding polytopal realisation is the Stokes polytope (or any member of the accordiohedron family except associahedron).
\end{itemize}
A potential stumbling block which provides hints for defining the world-sheet forms is the following : The CHY integrand for quartic interactions as defined in \cite{Baadsgaard:2015ifa} is defined on the $n-3$ dimensional moduli space which is diffeomorphic to kinematic space associahedron. Then why does this associahedron not suffice to obtain $\frac{n-4}{2}$ forms which define scattering amplitude for quartic interactions. 

Using seminal results from \cite{1807ppp,1906ppp}, we show that, rather strikingly, certain projective $\frac{n-4}{2}$ forms on kinematic space \emph{associahedra} define planar amplitudes for massless quartic interactions. The notion of projectivity is the same as the one introduced in \cite{Arkani-Hamed:2017mur} and is accomplished in this case by the dependence of these forms on certain quadrangulations $Q$ of an $n$-gon.  
For an $n$ particle scattering and reference quadrangulation $Q$ we denote these forms as $\Omega^{Q}_{n}$.  
Using results from \cite{1906ppp} we show that for any $n$, $\Omega^{Q}_{n}$ are canonical forms on $\frac{n-4}{2}$ dimensional Stokes polytopes that can always be realised inside kinematic space associahedra defined in \cite{Arkani-Hamed:2017mur,HughThomas}. For a special class of quadrangulations $Q$ which consists of parallel diagonals of an $n$-gon, the corresponding Stokes polytope has a hyper-cube topology and as we show,  there is a canonical realisation of it as a boundary of a higher dimensional associahedra. 

This understanding of quartic scattering amplitude aids us in writing down world-sheet formulae for such theories even though there is no known moduli space associated to Stokes polytopes. We define $\frac{n-4}{2}$ forms $\widehat{\Omega}^{Q}_{n}$ on ${\cal M}_{0,n}({\bf R})$ which are mapped (by scattering equations) to the  $\frac{n-4}{2}$ forms $\Omega^{Q}_{n}$ on kinematic space associahedron 
 For $\phi^{3}$ interactions these pushforward maps are the CHY integral formula. We attempt to re-write the pushforward maps for the lower forms as integral formulae for certain sub-manifolds in the moduli space. We show that this is indeed true for  forms which are labelled by quadrangulations that consist of parallel diagonals of an $n$-gon and the corresponding sub-manifold in the moduli space is precisely the image of kinematic space Stokes polytope under CHY scattering equations. We also argue that for other quadrangulations,  although such sub-manifolds can be defined for which the pushforward of $\widehat{\Omega}^{Q}_{n}$ can be written as an integral formula, the sub-manifolds are not diffeomorphic to kinematic space Stokes polytopes via scattering equations.

The paper is organised as follows. In Section \ref{aquickreview} we give a quick review of the core ideas contained in \cite{Arkani-Hamed:2017mur,Banerjee:2018tun} which are needed for the rest of the paper. In Section \ref{convexrealisationstokesposet} we use the beautiful ideas of \cite{1906ppp} to realise Stokes polytopes as convex polytopes in kinematic space. We also show that the same idea can be used to obtain convex realisation of associahedra and we obtain a subset of the associahedra realisation that were defined in \cite{Arkani-Hamed:2017mur,HughThomas}. This family of realisations which are all linearly diffeomorphic to each other are referred to as ABHY associahedra in \cite{HughThomas}. In Section \ref{phi4fl} we obtain a rather surprising result that starting with a so-called planar scattering form for quartic interactions which is projective $\Omega^{Q}_{n}$, there is at least one ABHY associahedron such that the restriction of $\Omega^{Q}_{n}$ as an $\frac{n-4}{2}$ form onto the associahedron yields a form that is proportional to the scattering amplitude $m^{Q}_{n}$. In Section \ref{stinass} we place the results of \cite{Banerjee:2018tun} on a rigorous footing by showing that $\Omega^{Q}_{n}$
are canonical top forms on Stokes polytopes which are always realised in one of the ABHY associahedra. In Section \ref{wsforquar} we use this result to show that on the (real section of) CHY moduli space ${\cal M}_{0,n}({\bf R})$ one can define world-sheet forms $\widehat{\Omega}^{Q}_{n}$ that are labelled by reference quadrangulation $Q$ such that their push-forward via CHY scattering equations produces the amplitude $m_{n}^{Q}$. In Section \ref{wsst1} we attempt to write these pushforward maps as ``CHY-type" integral formula on $\frac{n-4}{2}$ dimensional sub-manifolds of the moduli space that we call world-sheet Stokes geometries $\widetilde{{\cal S}}^{Q}_n$. If the reference quadrangulation $Q$ consists of non-intersecting diagonals then we show that $\widetilde{{\cal S}}^{Q}_n$ is diffeomorphic to convexly realised ${\cal S}_{n}^{Q}$ with diffeomorphisms generated by CHY scattering equations. For any other $Q$ we argue that although such world-sheet Stokes geometries exist, they are not diffeomorphic to kinematic space Stokes polytope via scattering equations. We end with an outlook. 

\section{A quick review of Positive geometries in Kinematic space}
\label{aquickreview}
In this section we review key aspects of kinematic space associahedron and associated canonical form. We also review the main results from \cite{Banerjee:2018tun} where the positive geometry for planar quartic interactions was discussed.  We encourage the interested reader to consult the original references as well as \cite{Raman:2019utu,Jagadale:2019byr} for details. 

In \cite{Arkani-Hamed:2017mur} it was shown that the tree-level planar amplitude for massless $\phi^3$ theory can be obtained from  a positive geometry known as the associahedron \cite{Stasheff1,Stasheff2} that sits inside the kinematic space. A positive geometry $\mathcal{A}$ is a closed geometry with boundaries of \emph{all} co-dimensions and has a unique differential form $\Omega(\mathcal{A}$), known as the \emph{canonical form} - a complex differential form defined by following properties :
\begin{enumerate}
	\item it has logarithmic singularities  at the boundaries of $\mathcal{A}$, 
	\item its singularities are recursive, \emph{i.e.} at every boundary $\mathcal{B}$, \   $ Res_{\mathcal{B}} \, \Omega (\mathcal{A})= \Omega(\mathcal{B})$,
	\item when ${\cal A}$ is a point $\Omega({\cal A}):=\pm 1$.
\end{enumerate}


\noindent For massless scalar $\phi^3$ theory, the positive geometry for an $n$ particle scattering is an $n-3$ dimensional polytope known as associahedron and denoted as $\mathcal{A}_n$. The vertices  of the associahedron correspond to complete triangulations and co-dimesion $k$-faces represent $k$-partial triangulations\footnote{A partial triangulation of regular $n$-gon is a set of non-crossing diagonals.} of the $n$-gon. 



To analyse planar amplitudes one uses planar kinematic variables :
\begin{align}
\label{planarkinvars}
X_{i j} &=  (p_i + p_{i+1} + \ldots + p_{j-1})^2 ; \quad 1\le i < j \le n \,.
\end{align}  
The Mandelstam variables can be expressed in terms of these as  :
\begin{align}
s_{ij} &= 2 p_i . p_j = X_{i,j+1} + X_{i+1,j} - X_{ij} - X_{i+1,j+1}\,.
\end{align}
%
%


In kinematic space $\mathcal{K}_n$, \emph{projectivity}\footnote{Projectivity implies $\Omega(\mathcal{A})$ can only be a function of ratios of Mandelstam variables.} uniquely defines the planar kinematic form  $\Omega^{(n-3)}_n$ up to an overall sign that can be safely ignored :
\begin{equation}
\Omega^{(n-3)}_n := \displaystyle \sum_{\text{planar  graph } g} \text{sign}(g)\bigwedge_{a=1}^{n-3} d \log X_{i_a,j_a}\,.
\end{equation}

To obtain the planar amplitude, as was  prescribed in \cite{Arkani-Hamed:2017mur}, one needs to embed the  $(n-3)$ dimensional $\mathcal{A}_n$ inside the $\frac{n(n-3)}{2}$ dimensional $\mathcal{K}_n$ and then pullback $\Omega^{(n-3)}_n$  on to  $\mathcal{A}_n$ to read off the amplitude. To embed $\mathcal{A}_n$ in $\mathcal{K}_n$ the following constraints \cite{Arkani-Hamed:2017mur} are imposed :
\begin{align}\label{eq:nima-asso}
X_{ij} &\ge 0; \quad 1 \le i < j \le n\\
s_{ij} &= - c_{ij}; \quad 1 \le i < j \le n-1, \quad |i-j| \ge 2\,.
\end{align}
The first set of constraints which are inequalities cuts out a simplex ($\Delta_n$) inside $\mathcal{K}_n$ without reducing the dimensionality. The second set of constraints are $\frac{(n-2)(n-3)}{2}$ in numbers and select out an $(n-3)$-dimensional sub-manifold  ($H_n$) inside $\mathcal{K}_n$. Thus the associahedron in $\mathcal{K}_n$ is :
\begin{align}
\mathcal{A}_n := \Delta_n \cap H_n\,.
\end{align}
Notice that the above set of constraints is not unique. It is just one way of realising the associahedron in the positive region of kinematic space -- a particular convex realisation discovered in \cite{Arkani-Hamed:2017mur} which is a special class of the more general realisations known as ABHY associahedra \footnote{We will need these more general realisation of associahedra to connect the canonical form for $\phi^4$ theory in kinematic space to the differential form living in the subspace of world-sheet moduli space (See Section \ref{phi4fl}).} .


In \cite{Banerjee:2018tun}, the program of obtaining the amplitude from positive geometry was extended to $\phi^4$ interaction (and to other polynomial interactions in \cite{Raman:2019utu, Jagadale:2019byr}) for a few low values of $n$. To make this article self-contained and to set up notation, we briefly summarize the main findings of \cite{Banerjee:2018tun}. For $\phi^4$ theory, one should consider quadrangulations of an $n$-gon, where $n$ is always even.  The number of ways one can completely quadrangulate an $n=( 2N+2)$-gon is given by the famous Fuss-Catalan number, $F_{N}=\frac{1}{2N+1} \binom {3N}{N}.$ The number of internal lines for an $n$-particle graph which is the same as the number of diagonals in the quadrangulation of an $n$-gon is given by $N-1=\frac{n-4}{2}$. To connect to the positive geometry framework one should look for a combinatorial polytope living in $\frac{n-4}{2}$ dimensions and which has $F_{N}$ number of vertices. Any polytope in one dimension is a line segment. For the simplest non-trivial case  $n=6$ (\emph{i.e.} $N=2$), although the polytope is one dimensional, it would have $3$ vertices. To obtain a positive geometry one needs to systematically drop one of these vertices (\emph{i.e.} complete quadrangulations). The definition of the compatibility of a diagonal with a reference quadrangulation , which we will give now, does precisely this. 

The compatibility of a diagonal with a reference quadrangulation is defined as follows. Consider a regular polygon $P$ with $n$ vertices and label its vertices with $1,2,\ldots, n$. We will call this polygon, solid polygon and we will call its diagonals, solid diagonals. Now, consider the regular polygon formed by the mid-points of the sides of the solid polygon and label the mid-points of the sides $ [ i,i+1 ]$ by $i'$. We will call this polygon, hollow polygon and we will call its diagonals, hollow diagonals. Given a reference quadrangulation of the solid polygon $Q$ a cut $C(i',j')$ of the hollow diagonal $(i',j')$ is a set comprising the sides $[i,i+1]$ and $[j,j+1]$ of the solid polygon along with the diagonals of the quadrangulation $Q$ which intersect the diagonal $(i',j')$. We say the solid diagonal $(i,j)$ is compatible with the quadrangulation $Q$ if the cut $C(i',j')$ is connected. A quadrangulation $Q'$ is compatible with reference quadrangulation $Q$ if and only if all the diagonals of the quadrangulation $Q'$ are compatible with reference quadrangulation $Q$. This notion of compatibility is same as $Q$-compatiblity defined in \cite{Baryshnikov}. The vertices of the Stokes polytope with reference quadrangulation $Q$ are the quadrangulations that are compatible with $Q$. Note that the definition of Stokes polytope depends on the reference quadrangulation and different quadrangulations will give you different Stokes polytope. 



One has to introduce a differential $\frac{n-4}{2}$-form in the kinematic space. Using compatibility we can define a new operation on the $n$-gon called \emph{Flip}. Any $n$-point diagram with $n \ge 6$ will have one or more hexagons inside it. Flip is an operation of replacing  a diagonal of any such hexagon inside the quadrangulation of the polygon with a diagonal compatible with the reference quadrangulation. Flip helps assign particular signs ($\sigma = \pm 1$) to each compatible quadrangulation relative to its reference quadrangulation.

Let $Q$ be a quadrangulation of an $n$-gon which is associated to an planar Feynman diagram with propagators given by $X_{i_1},\dots, X_{i_\frac{n-4}{2}}$. We define the ($Q$-dependent) planar scattering form, 
\begin{equation}
\label{projectiveplanarform-Qdep}
\Omega^{Q}_{n}\ = \displaystyle \, \sum_{graphs} (-1)^{\sigma(\textrm{flip})} \, d \ln X_{i_{1}}\bigwedge  \, d \ln X_{i_{2}} \dots \, \bigwedge d \ln X_{i_{\frac{n-4}{2}}}
\end{equation}
where $\sigma(\textrm{flip})\ =\, \pm 1$ depending on whether the quadrangulation $X_{i_{1}},\dots, X_{i_{\frac{n-4}{2}}}$ can be obtained from $Q$ by even or odd number of flips.

Given $Q$ the differential form does not contain contribution from all the $\phi^4$ propagators. Consider the simplest case, i.e. $n=6$. Let's start with $Q = 14$. The set of $Q$ compatible quadrangulations are $\{\,(14\ , +),\ (36\ , -) \}$. We have attached a sign to each of the quadrangulation which measures the number of flips needed to reach it starting from reference $Q = 14$. The form $\Omega^{Q}_{6}$ on the kinematic space is given by,
\begin{equation} \label{eq:planar-phi4-form-Q14}
\Omega^{Q=14}_{6}\ =\  ( d\ln X_{14}\ -\ d\ln X_{36})\,. 
\end{equation}
It is evident that it does not capture the singularity associated to $X_{25}$ channel. We can get around this problem by considering other possible $Q$s whose forms on kinematic space are given by, 
\begin{align}
\label{eq:planar-phi4-form-Q2536}
\Omega^{Q=36}_{6}&=(d\ln X_{36}-d\ln X_{25})\cr
\Omega^{Q=25}_{6}&=(d\ln X_{25}-d\ln X_{14})\,.
\end{align}
In fact these differential forms naturally descend to the  canonical forms on Stokes polytopes ${\cal S}_{n}^{Q}$. Since a Stokes polytope is a positive geometry, it has a canonical form associated to it which has  (logarithmic) singularities on all the facets, such that the residue of restriction of this form on any of the facet equals the canonical form on the facet.
Once we have the Stokes polytopes and their corresponding canonical forms we need to embed the polytopes inside the positive region of  kinematic space. To obtain the planar amplitude we pull back the canonical forms on the polytopes. 
There are several convex realisations of Stokes polytopes. Their realisation as a simple polytope is given in \cite{Baryshnikov, Chapoton}.  In \cite{Banerjee:2018tun} inspired by ideas outlined in \cite{Arkani-Hamed:2017mur} a few low dimensional Stokes polytopes (namely $n=6, 8$ and $10$) were given particular convex realisation in the kinematic space. The strategy was to embed the Stokes polytopes ${\cal S}_{n}^{Q}$ with dimension  $\frac{n-4}{2}$ inside corresponding associahedra ($\mathcal{A}_n$) with dimension $n-3$, for given number of particle $n$. We illustrate these ideas through a simple example.

For $n=6$ if we start with reference $Q_{1} = (14)$,  the $Q_{1}$ compatible set is given by $\{(14, +),\ (36,-)\}$. The corresponding Stokes polytope is one dimensional with two vertices. We locate this Stokes polytope inside the kinematic space via the following constraints. 
\begin{equation}\label{n=6constraint1}
\begin{array}{lll}
s_{ij}\ =\ - c_{ij} \quad  \forall\ 1 \le i  <  j  \le n-1=5 , \ |i-j| \geq 2 \\
X_{13}\ =\ d_{13},\ X_{15}\ =\ d_{15}, \textrm{with}\ d_{13},\ d_{15}\ >\ 0
\end{array}
\end{equation}
The first line of constraints are precisely the ones which define the three dimensional kinematic associahedron ${\cal A}_{6}$ inside ${\cal K}_{6}$. The other two extra constraints have been imposed to make it a one-dimensional Stokes polytope ${\cal S}^{\{14\}}_{6}$ sitting inside ${\cal A}_{6}$. From the perspective of Feynman diagrams, these constraints are rather natural as planar variables from this set can never occur in Feynman diagrams of $\phi^{4}$ theory. 

Using the above constraints, it can be checked that the planar kinematic variables satisfy,
\begin{align}
\label{eq:const-X14}
X_{36} + X_{14}&=c_{14}+c_{24}+c_{15}+c_{25}=\displaystyle \sum_{i=1}^{2} \sum_{j=4}^{5} c_{i j} \geq 0, \cr
X_{25}&> 0
\end{align}
It is interesting to note that these constants on the RHS of equation \eqref{eq:const-X14} are independent of $d_{ij}$'s. As we will see in Section (\ref{phi4fl}), this observation is a corollary of a theorem which will have rather serious ramifications for us.  
We can now pull back the canonical form $\Omega^{\{14\}}_{6}$ on ${\cal S}^{\{14\}}_{6}$

\begin{equation}\label{spform2}
\begin{array}{lll}
\omega^{Q_{1}}_{6}=\left(\frac{1}{X_{14}}+\frac{1}{X_{36}}\right) d X_{14}
=:\ m_{6}({\cal S}^{Q_{1}}_{6})\ d X_{14}
\end{array}
\end{equation}

$m_{6}({\cal S}^{Q_{1}}_{6})$ is the canonical rational function associated to the Stokes polytope ${\cal S}^{Q_{1}}_{6}$. This is just a partial planar amplitude. To obtain the full planar amplitude we need to consider contributions from the other two Stokes polytopes (with $Q = 36$ and $25$). For $Q_{2} = (25)$  the $Q_{2}$-compatible set is given by $\{(25, +),\ (14,-)\}$ and for  $Q_{3} = (36)$  the $Q_{3}$-compatible set is $\{(36, +),\ (25,-)\}$. 

For the full planar amplitude one adds  up the canonical rational functions associated to all different Stokes polytopes with particular weights,
 \begin{align} 
\widetilde{{\cal M}}_{6} &:= \alpha_{Q_{1}} \, m_{6}({\cal S}^{Q_{1}}_{6})\ + \alpha_{Q_{2}} \, m_{6}({\cal S}^{Q_{2}}_{6})\ + \alpha_{Q_{3}} \, m_{6}({\cal S}^{Q_{3}}_{6})\ \\
&= \alpha_{Q_{1}} \left(\frac{1}{X_{14}}\ + \frac{1}{X_{36}}\right)\ +\ \alpha_{Q_{2}}\left(\frac{1}{X_{25}}\ +\ \frac{1}{X_{14}}\right)\ +\ \alpha_{Q_{3}} \left(\frac{1}{X_{36}}\ +\ \frac{1}{X_{25}}\right).
\end{align} 
It is evident that $\widetilde{{\cal M}}_{6}={\cal M}_{6}$ (the planar amplitude) if and only if $\alpha_{Q_{1}} = \alpha_{Q_{2}}= \alpha_{Q_{3}} =  \frac{1}{2}$ . 

Higher point amplitudes can also be obtained in a similar way since thanks to \cite{1906ppp} we have a convex realisation of Stokes polytopes inside ABHY associahedra for arbitrary $n$. 
To obtain the $n$-point planar amplitude for $\phi^4$ we need to consider the rational function contributions from all possible Stokes polytopes and sum them up with particular weights,
\begin{equation}\label{weights1}
\widetilde{{\cal M}}_{n}\ =\ \sum_{Q} \alpha_{Q} \ m_{n}(Q).
\end{equation}

The weights $\alpha_{Q}$ may look like arbitrary free parameters but they are constrained by physical factorization of amplitudes which is in one to one correspondence with combinatorial factorization of Stokes polytopes. Given any diagonal $(ij)$, the corresponding facet $X_{ij} = 0$ is a product of two lower dimensional Stokes polytopes, 
\begin{equation}\label{factor1}
{\cal S}_{n}^{Q}\bigg|_{X_{ij} = 0}\ \equiv {\cal S}_{m}^{Q_{1}} \times {\cal S}_{n+2-m}^{Q_{2}}
\end{equation}
where $Q_{1}$ and $Q_{2}$ are such that $Q_{1}  \cup Q_{2} \cup (ij) = Q$. It can also be seen that the residue over each Stokes polytope which contains a boundary $X_{ij} \rightarrow 0$ factorizes into residues over lower dimensional Stokes polytopes. 
This factorization property naturally implies factorization of amplitudes as follows. Consider the $n$-gon with a diagonal $(ij)$ (with $i,j$ such that this diagonal can be part of a quadrangulation). This diagonal subdivides the $n$-gon into a two polygons with vertices $L := \{i,\dots,j\}$ and $R := \{j,\dots,n,1,\dots i\}$ (see figure  \ref{fig:factorization}). 
\begin{figure}
	\centering 
	\includegraphics[width=8cm]{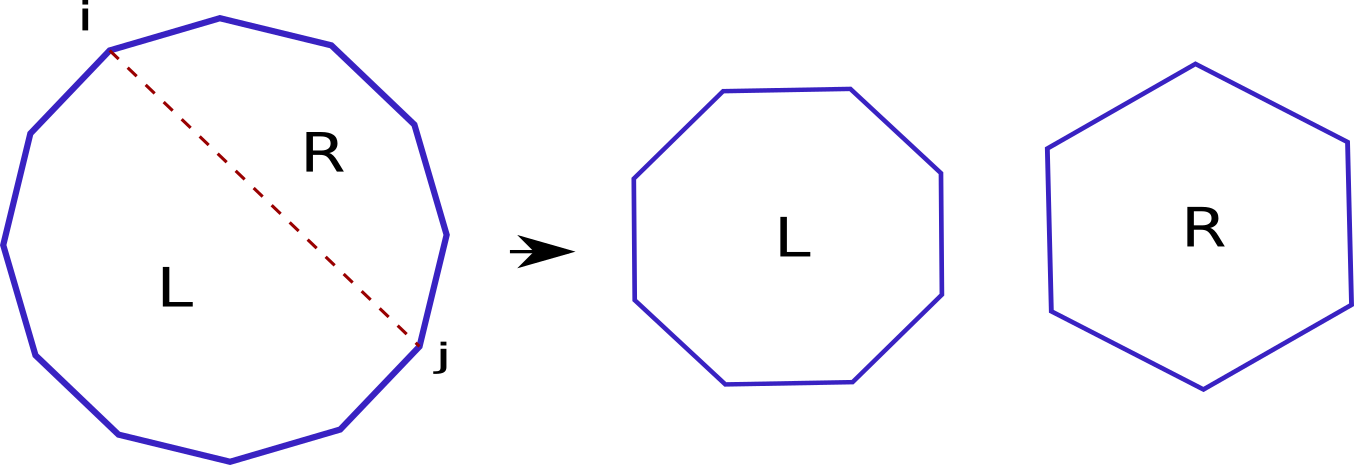} 
	\caption{\label{fig:factorization} A dual polygon breaks into two smaller polygons when a $\phi^4$ propagator is taken on-shell. In this example,  a 12-gon breaks into an octagon and a hexagon.}
\end{figure}
These two polygons, by construction,  correspond to two Stokes polytopes $\widetilde{M}_{\vert j - i + 1\vert}$ $=: \widetilde{{\cal M}}_{L}$ and $\ \widetilde{M}_{n+2 - (\vert j - i + 1\vert)}$ $=:\widetilde{{\cal M}}_{R}$ respectively. This implies if one takes $X_{ij} \to 0$ the amplitude factorizes to two lower point ones,  
\begin{equation} \label{factor2}
\widetilde{{\cal M}}_{n}\bigg|_{X_{ij}\, =\, 0}\ =\ \widetilde{{\cal M}}_{L} \,\frac{1}{X_{ij}}\, \widetilde{{\cal M}}_{R}
\end{equation}
This in turn put constraints on $\alpha_{Q}$'s which  have to satisfy the following relation, 
\begin{equation}\label{eq:alpha-cons}
\sum_{Q\ \textrm{containing} (ij)} \alpha_{Q}\ =\ \sum_{Q_{L}, Q_{R}}\alpha_{Q_{L}}\alpha_{Q_{R}}
\end{equation}
where $Q_{L}$ and $Q_{R}$ range over all the quadrangulations of the two polygons $L$ and $R$.


\section{Convex realisation of Stokes Polytope in the Kinematic Space} 
\label{convexrealisationstokesposet}
Convex realisations of accordiohedra provide us with positive geometries whose canonical forms yield scattering amplitudes for planar scalar field theories with polynomial interactions \cite{Arkani-Hamed:2017mur, Banerjee:2018tun, Raman:2019utu, Jagadale:2019byr}. For cubic interactions, such a realisation of associahedron was derived in \cite{Arkani-Hamed:2017mur}. In the context of quartic interactions, based on certain physics inputs motivated by the ideas in \cite{Arkani-Hamed:2017mur} an attempt was made in \cite{Banerjee:2018tun} to define a convex realisation of Stokes polytopes in kinematic space. In contrast to the associahedron case, this construction did not provide a \emph{canonical} convex realisation of an arbitrary Stokes polytope of generic dimension, and only a few examples of low dimensional Stokes polytopes were given. A rigorous and complete realisation of Stokes polytopes as embedded in kinematic space was provided in \cite{1906ppp}.  For sake of pedagogy and to ease the task of fellow travellers who venture in these worlds, we will now give a detailed transcription of the state of the art construction developed in \cite{1906ppp} to a working physicist's language.  

Although the construction in \cite{1906ppp} relies on inputs from theory of quivers and corresponding path algebras \cite{1807ppp}, the convex realisation of Stokes polytope can be understood without a detailed knowledge of them. In the following we state the key result in \cite{1906ppp} pertaining to such realisations (theorem 2.44 in \cite{1906ppp}) and explain the construction ``algorithmically" . 

\begin{itemize}
\item We consider an $n$-gon with ordered vertices which go clockwise. We label the external edges such that the one that connects vertices $i$ and $i+1$ is $[i,i+1]$. We denote by $\cal E$ the set of all edges. We consider a reference quadrangulation $Q$ that consists of $\frac{n-4}{2}$ diagonals, which divide the $n$-gon into quadrilaterals. We refer to these quadrilaterals as cells and use $\cal C$ to denote them. The diagonals in $Q$ are labelled such that if one of them intersects vertices $p$ and $q$ with $p<q$ we denote it as $D_{p,q}$. We then consider those edges of the $n$-gon that do not intersect any of the diagonals and refer to the set of such edges as $Q^{c}$. We place a ``hollow" vertex on each edge $[i,i+1]\in{\cal E}\setminus Q^{c}$ and label it $\widetilde{v}_{i}$.  We also place a hollow vertex on every $D_{pq}\in Q$ and label it $v_{pq}$.
See Figure.~\ref{figone} for an example. 
\begin{figure}[H]
    \centering
    \includegraphics[scale=0.3]{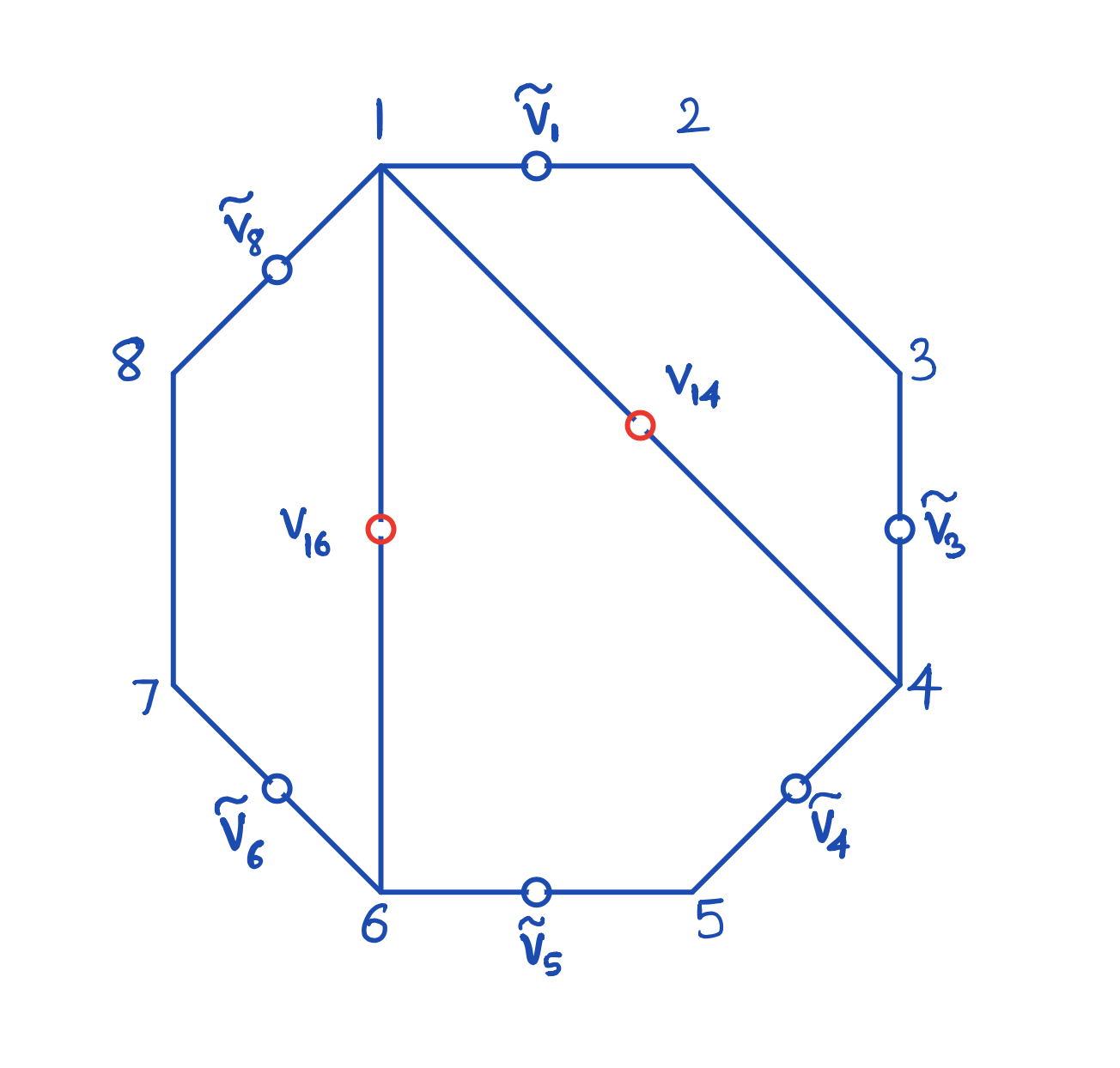}
    \caption{\{14,16\} quadrangulation of the octagon and the hollow vertices on the external edges and the diagonals.}
                \label{figone}
\end{figure}
%
%
\item In each cell we draw line segments such that,
\begin{itemize}
\item they have one end point on the hollow vertex on an external edge and the other on the hollow vertex on a diagonal which intersects the external edge\,,
\item they have their end points on two diagonals if the two intersect at a vertex of the $n$-gon\,.
\end{itemize}
We place arrows on these line segments such that within each cell ${\cal C}$, a path $p_{\cal C}$ comprising of such arrows, which begins on an external edge and ends on another is oriented counter-clockwise. See Figure.~\ref{fig2} for an example. It is easy to see that within each cell such a path is unique.
%
\begin{figure}
    \centering
    \includegraphics[scale=0.3]{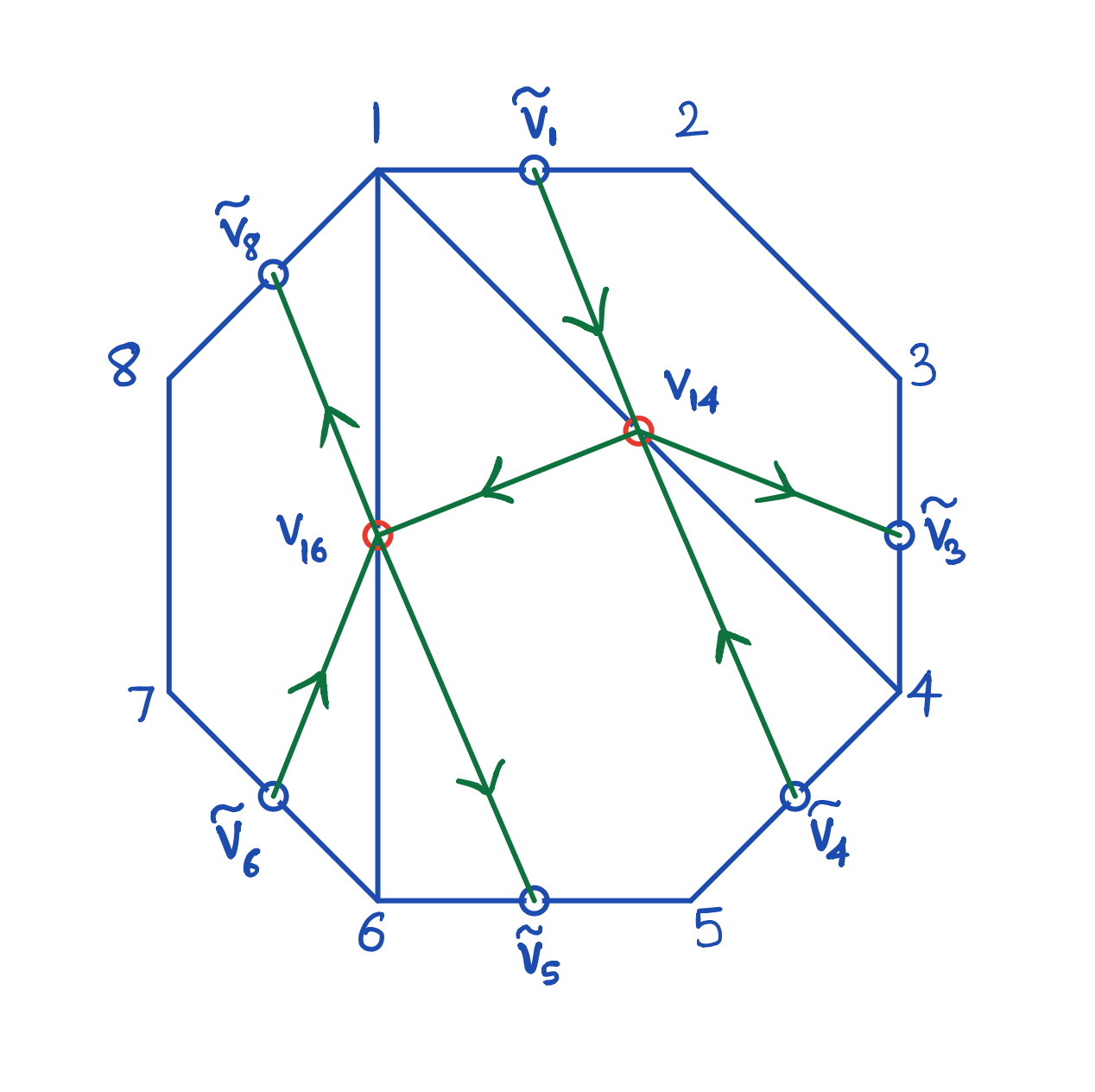}
    \caption{Paths within the three cells are marked in green.}
    \label{fig2}
\end{figure}
\item We then consider the set of all possible internal paths $p_{\{v_{i_{k},j_{k},\dots,v_{i_{\ell},j_{\ell}}}\}}$ which begin on $v_{i_{k},j_{k}}$ and end on $v_{i_{\ell},j_{\ell}}$ without intersecting any external edge. 
Once again it is clear from the figure that such paths are unique. Note that the set includes paths of length zero ($p_{v_{i,j}}$ for some $v_{i,j}$) which are the hollow vertices on diagonals themselves. The cardinality of this set denoted as $K$ depends on the choice of reference quadrangulation.
%
%

We now define the set of proper walks. A proper walk is a path between two hollow vertices on the external edges such that it does not contain two segments in the same cell. We denote the set of such proper walks by ${\cal W}$. 
\item Given a path $p$ we define a subset of proper walks that are  the key players in the construction. 
\begin{itemize}
\item  A proper walk containing $p$, $w_{p}$ reaches peak at $p$ if both the arrows on $w_{p}$ incident on $p$ are outgoing. 
A walk $w_{p}$ passing through $p$ which has a peak  $p$ and is such that all arrows not on $p$ and not incident on $p$ are incoming is denoted as $\hook p \unhook$.
\item Similarly, $w_{p}$ reaches a dip at $p$ if the two arrows on it which are incident on $p$ are incoming. If $w_{p}$ is such that all arrows not on $p$ and not incident on $p$ are outgoing, it is denoted as $\cohook p \counhook$.
\item We call $w_{p}$ neutral if one of the incident arrows on $p$ is incoming and the other is outgoing with incoming (outgoing) arrow preceded (followed) by arrows with reverse orientation. We denote the two neutral paths incident on $p$ as $\hook p \counhook$ and $ \cohook p \unhook$ respectively. 
\end{itemize}
\item It was shown in \cite{1906ppp} that for each path $p$, there is a unique proper walk $w_{p}$ which has a peak (or a dip) at $p$ and there are precisely two neutral walks at $p$. 
\end{itemize}
With the above algorithm in hand, we can now paraphrase  in physicist's language the main theorem in \cite{1906ppp} concerning convex realisation of Stokes polytopes as follows :

In the positive region of kinematic space and for any choice of constants $d_{p_{1}},\dots, d_{p_{K}}$, the following constraints locate Stokes polytope :
\begin{align}\label{stokc}
X_{\hook p_{i} \unhook}+X_{\cohook p_{i} \counhook}-X_{\hook p_{i} \counhook}-X_{ \cohook p_{i} \unhook}=d_{p_{i}},\quad
\forall\,\, i\in\{1,\dots, K\}\,.
\end{align}
This result provides a convex realisation of Stokes polytopes (and in general any accordiohedron) inside the positive hyper-quadrant of the kinematic space. We work out a simple example of this algorithm in Appendix \ref{appenalgo}.

We will now show that the convex realisation of any accordiohedron discussed above  includes a subset of (generalised) associahedra defined in \cite{Arkani-Hamed:2017mur, HughThomas}. 
\subsection{Convex realisations of associahedra}\label{sec3}
Unlike Stokes polytopes (and other accordiohedra based on $p$-angulations with $p>3$) there is a unique combinatorial associahedron associated to partial triangulations of a polygon.  However, there are several inequivalent convex realisations of an associahedron in an ambient Euclidean space. A special class of these realisations that correspond to the triangulation 
$T=\{(1,3),(1,4),\dots,(1,n-1)\}$ of the $n$-gon was discovered in \cite{Arkani-Hamed:2017mur} and we denote them by ${\cal A}_{n}$. We denote realisations based on a triangulation $T$ other than the one mentioned above by ${\cal A}_{n}^{T}$. These realisations were discussed in \cite{1906ppp} and analysed as a sub-set of a wider class of generalised associahedra in \cite{HughThomas}. Following the analysis of \cite{HughThomas} we will refer to ${\cal A}_{n}^{T}$ as ABHY associahedra. However we note that ${\cal A}_{n}^{T}$ are only a sub-set of convex realisations discussed in \cite{HughThomas}, namely those associated to linear quivers with arbitrary ordering \footnote{The easiest way to see this is via the comparison of so-called vertices and arrows of the linear quiver in \cite{HughThomas} with hollow vertices and arrows on $T$ that generate ${\cal A}_{n}^{T}$. A detailed discussion on this, though straightforward will cause too much digression and we avoid it in the interest of pedagogy.}\,\footnote{We are indebted to Song He for pointing out this injection  between $\{{\cal A}_{n}^{T}\}$ and ABHY associahedra to us.} .

\begin{claim}
For a reference triangulation $T$ of an $n$-gon, the algorithm described in Section.~\ref{convexrealisationstokesposet} gives rise to a convex realisation of the $n-3$ dimensional associahedron ${\cal A}_{n}^{T}$ in the $\frac{n(n-3)}{2}$ dimensional planar kinematic space (with basis given by planar kinematic variables) with the constraints :
\begin{equation}\label{sijcijT} 
s_{ij}=-c_{ij}\quad\forall\,\, (ij)\notin T^{c}\quad\text{and}\quad |i-j|\ge 2\,,
\end{equation}
where $T^{c}$ is the triangulation obtained via a $\frac{2\pi}{n}$ rotation of each diagonal in $T$ in the counterclockwise sense, i.e. if $(k,\ell)\in T$ then $(k-1,\ell-1)\in T^{c}$. 
\end{claim}
%
\begin{proof} 
Consider the path $p$ that starts on $D_{i_k j_k}$ and ends on $D_{i_\ell j_\ell}$. Both $D_{i_kj_k}$ and $D_{i_\ell j_\ell}$ belong to two triangular cells, one which contains an arrow of $p$ and one which does not. Let $k$ be the third vertex of the triangular cell adjacent to $D_{i_k j_k}$ and does not contain an arrow of $p$. Similarly, let $\ell$ be the third vertex of the triangular cell adjacent to $D_{i_\ell j_\ell}$ and does not contain an arrow of $p$. See Figure.~\ref{fig3}.
\begin{figure}
    \centering
    \includegraphics[scale=0.3]{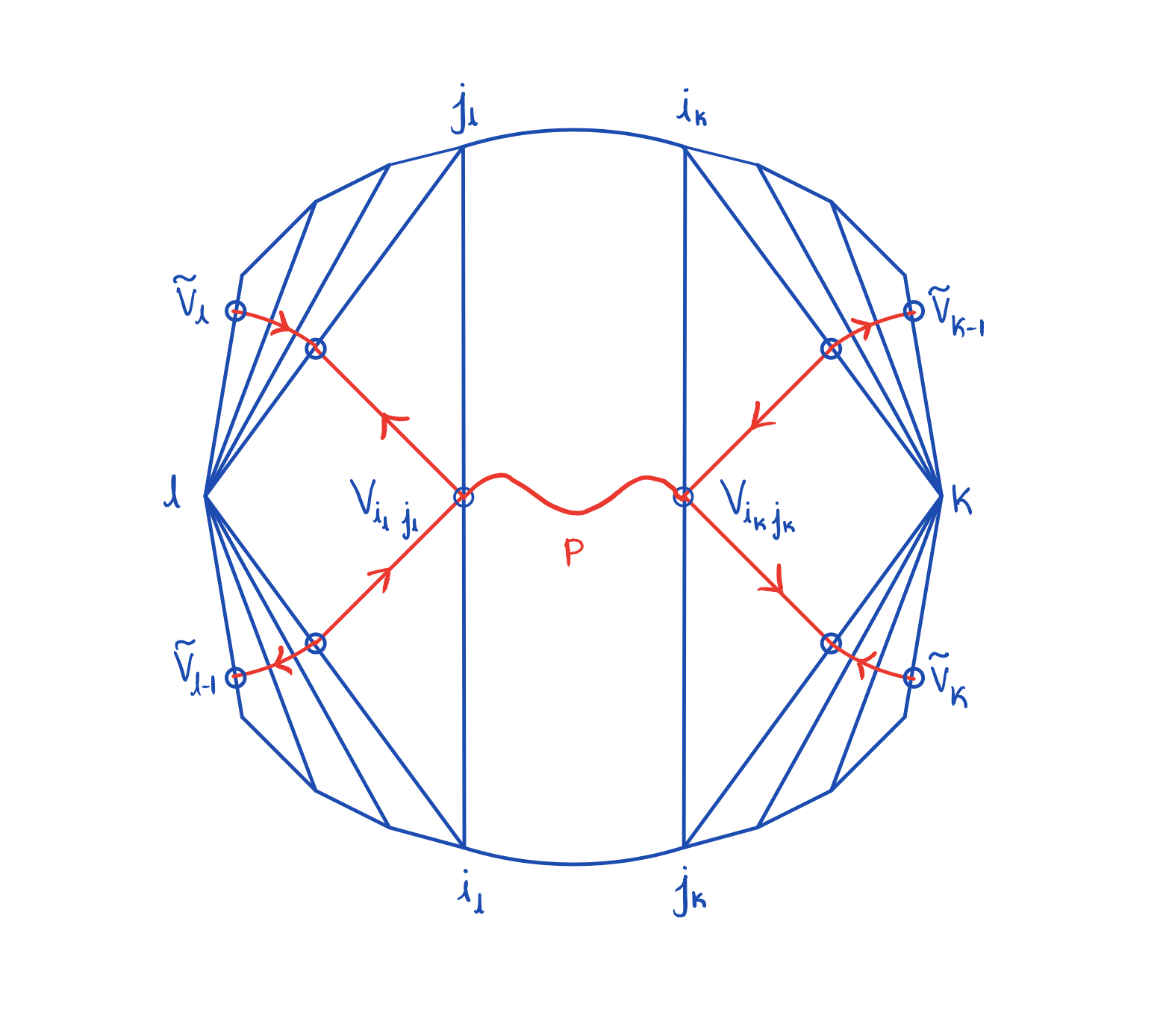}
    \caption{\,}
    \label{fig3}
\end{figure}
\noindent The outgoing arrow from $v_{i_k,j_k}$ goes to the diagonal/edge $D_{k,j_k}$ and the subsequent incoming arrows go to the hollow vertex $\widetilde{v}_{k}$. The incoming arrow from $v_{i_k,j_k}$ goes to the diagonal/edge $D_{i_k,k}$ and the subsequent outgoing arrows go to the hollow vertex $\widetilde{v}_{k-1}$. A similar analysis on the $\ell$ side tells us $ X_{\hook p \unhook} = X_{k,\ell}$, $X_{\cohook p \counhook}= X_{k-1,\ell-1}$, $X_{\hook p \counhook}=X_{k,\ell-1}$, and $X_{\cohook p \unhook }= X_{k-1,\ell}$. Thus we have :
\begin{align}
\label{intermediateAssociahedronconstraint}
X_{\hook p \unhook}+X_{\cohook p \counhook}-X_{\hook p \counhook}-X_{ \cohook p \unhook}&=
X_{k,\ell} + X_{k-1,\ell-1}-X_{k-1,\ell}-X_{k,\ell-1}\cr
&=-s_{k-1,\ell-1}\,,
\end{align}
where the second line follows from the definition of planar kinematic variables in \eqref{planarkinvars}.

Since the diagonal $D_{k,l}$ intersects the diagonals $D_{i_k,j_k}$ and $D_{i_\ell,j_\ell}$, it does not belong to the reference triangulation $T$. Thus $(k-1,\ell-1)\notin T^c$ and hence are constants according to \eqref{sijcijT}. Imposing this in \eqref{intermediateAssociahedronconstraint} we obtain :
\begin{align}\label{abhyppp}
X_{\hook p \unhook}+X_{\cohook p \counhook}-X_{\hook p \counhook}-X_{\cohook p \unhook}&=c_{k-1,\ell-1}\,,
\end{align}
which is precisely the associahedron constraint associated with the path $p$. Such an equation is valid for every diagonal which does not belong to the reference triangulation. This proves the claim.
\end{proof}{}

\section{Planar $\phi^{4}$ amplitudes from lower forms on ABHY associahedra}\label{phi4fl}
There is a happy relationship between the Stokes polytope constraints in equation \eqref{stokc} and the associahedron constraints in equation \eqref{abhyppp} if $Q\subset T$. In Appendix \ref{appb}, we show that the Stokes polytope constraints \eqref{stokc} are a linear combination of the ABHY associahedron constraints \eqref{abhyppp}.
This observation immediately implies that the restriction of the projective planar scattering form $\Omega^{Q}_{n}$ \eqref{projectiveplanarform-Qdep} onto the ABHY associahedron ${\cal A}^{T}_{n}$ (with $Q\subset T$) is an $\frac{n-4}{2}$ form proportional to the partial contribution to the amplitude $m^{Q}_{n}$.  

A rather striking consequence of this assertion  is that the convex realisation of Stokes polytope never enters the picture. The positive geometries which produce the amplitudes for $\phi^{3}$ as well as for $\phi^{4}$ massless planar interactions are always the ABHY associahedra. While amplitudes in bi-adjoint $\phi^{3}$ theory are given by the canonical top form on any one of the associahedra ${\cal A}_{n}^{T}$ (e.g. ${\cal A}_{n}$), the forms for $\phi^{4}$ theory are a set of lower $\frac{n-4}{2}$ forms on \{${\cal A}_{n}^{T}$\}. These forms are simply the restrictions of the projective planar scattering form $\Omega^{Q}_{n}$ for some reference quadrangulation $Q$. 

We first illustrate these ideas in the concrete example of 6 particle scattering. Let us first consider $Q=\{14\}$ and the ABHY associahedron ${\cal A}_{6}$ defined using equation~(\ref{abhyppp}) with $T=\{13,14,15\}$. Using associahedron constraints we get :
%
\begin{align}
X_{14}+X_{36}=\sum_{i=1}^{2}\sum_{j=4}^{5}\ c_{ij}\,.
\end{align}
We see that on ${\cal A}_{6}$ the planar scattering form $\Omega^{\{14\}}_{6}$ \eqref{eq:planar-phi4-form-Q14} 
descends to a one-form as :
\begin{align}
\Omega^{\{14\}}_{6}\bigg{|}_{{\cal A}_{6}} =\left(\frac{1}{X_{14}}+\frac{1}{X_{36}}\right)dX_{14}\,.
\end{align}
Let us now consider $Q=\{36\}$ and the associahedron ${\cal A}_{6}^{\{13,35,36\}}$. Associahedron constraints give :
%
\begin{align}
X_{36}+X_{25}=c_{13}+c_{14}+c_{36}+c_{46}\,.
\end{align}
Thus on ${\cal A}_{6}^{\{13,35,36\}}$ the planar scattering $\Omega^{\{36\}}_{6}$ in \eqref{eq:planar-phi4-form-Q2536} descends to :
\begin{equation}
\Omega^{\{36\}}_{6}\bigg{|}_{{\cal A}_{6}^{\{13,35,36\}}}=\left(\frac{1}{X_{36}}\ +\ \frac{1}{X_{25}}\right)dX_{36}\,.
\end{equation}
We now consider $Q=\{25\}$ and the associahedron ${\cal A}_{6}^{\{24,25,26\}}$. As before we can use associahedron constraints to show that :
\begin{equation}
X_{25}+X_{14}=\sum_{j=2}^{3}\sum_{j=5}^{6}\ c_{ij}\,.
\end{equation}
On ${\cal A}_{6}^{\{24,25,26\}}$ the planar scattering $\Omega^{\{25\}}_{6}$ in \eqref{eq:planar-phi4-form-Q2536} descends to :
\begin{align}
\Omega^{\{25\}}_{6}\bigg{|}_{{\cal A}_{6}^{\{24,25,26\}}}=\left(\frac{1}{X_{25}}+\frac{1}{X_{14}}\right)dX_{25}\,.
\end{align}
We thus see that the complete amplitude $m_{6}=\sum_{Q}\,\alpha_{Q}\ m^{Q}_{6}$ is obtained by summing over rational functions associated to $\Omega^{(Q)}_{6}\vert_{{\cal A}_{6}^{T}}$ with $Q\subset T$.  

It can be readily seen that this result generalises to the $n$ point amplitude. Since the constraints relating $X_{ij}$ such that $(ij)\in Q\cup Q^{m}$ \footnote{$Q^{m}$ is the set of all diagonals obtained from $Q$ by mutation.} are contained in the associahedron constraints when $Q\subset T$, we have the following result : 
All  ``partial" amplitudes $m^{Q}_{n}$ associated to quadrangulation $Q$ are simply restrictions of the projective planar scattering form $\Omega^{Q}_{n}$ to \emph{any} of the ABHY associahedra ${\cal A}_{n}^{T}$ with $Q\subset T$. 

We thus see that there are \emph{unique} projective forms on ABHY associahedra, such that a weighted sum over all of them produces the planar amplitude $m_{n}$ for quartic interactions. Combinatorics of a Stokes polytope is used to define projective forms but rather than viewing these forms as top forms on a convex polytope, we can view them as lower forms on a kinematic associahedron.

We find this perspective rather appealing as the only positive geometries needed to define amplitudes for $\phi^{p}$ theories are the ABHY associahedra. But rather than considering only canonical top forms on it which generate (planar) $\phi^{3}$ amplitude, we need to consider projective forms of various lower ranks which generate planar amplitudes for other scalar interactions. Although our analysis in this section and Appendix (\ref{appb}) are restricted to quartic interactions, it can be generalised to all polynomial interactions.

\section{Convex realisation of Stokes polytopes inside ABHY associahedra}\label{stinass}
In the last section we saw that projectivity and ABHY associahedra are enough to obtain the planar amplitudes in scalar field theory with quartic interactions. In the case of $\phi^{3}$ interactions, the projective form is the \emph{unique} canonical top form on kinematic space associahedron. Hence a natural question is if the lower dimensional projective forms $\Omega^{Q}_{n}$ can also be understood as canonical top forms on positive geometries in kinematic space. This idea was realised in \cite{Banerjee:2018tun,Raman:2019utu,Jagadale:2019byr} where convex realisations of Stokes polytopes (and accordiohedra in general) were used to restrict the planar scattering forms as unique canonical top forms on them and thereby extract the amplitude of various theories.

In this section we analyse such realisations in the context of Stokes polytopes. 
We denote such convex realisations as ${\cal S}_{n}^{Q}$. In \cite{Banerjee:2018tun} one, two and three dimensional Stokes polytopes were convexly realised inside the kinematic space. In fact the realisation was such that the polytopes sat either in the interior of, or on the boundary of ${\cal A}_{n}$.  We will now try to realise generic (in any dimension and for any quadrangulation $Q$) Stokes polytopes inside \footnote{By ``inside" we mean  either in the interior of the associahedron or embedded as one of the lower dimensional facets of the associahedron.}  ABHY associahedra ${\cal A}_{n}^{T}$. In Appendix (\ref{spoofQ}) we also prove that the canonical forms associated to ${\cal S}_{n}^{Q}$ indeed generate amplitudes for planar quartic interactions, a result that was anticipated in \cite{Banerjee:2018tun, Raman:2019utu}. 
Our primary result in this section is that any Stokes polytope with a reference quadrangulation $Q$ and ``size" determined by certain constants $d_{ij}$ can always be embedded inside (at least) one of the ABHY associahedra which satisfies the following properties :
\begin{itemize}
\item It is an ABHY associahedron ${\cal A}_n^{T}$ with reference triangulation $T\supset Q$. 
\item The constants $c_{ij}$ which determine the size of ${\cal A}_{n}^{T}$ in ${\cal K}_{n}$ bound $d_{ij}$ from above and below, as in e.g. equation \eqref{boundond}.
\end{itemize}
In the following we prove our claim for $Q=\{(1,4),(1,6),\dots,(1,n-2)\}$ where the ``enveloping" associahedron is ${\cal A}_{n}$ defined in \cite{Arkani-Hamed:2017mur}.

\subsection{Embedding ${\cal S}_{n}^{\{14,16,\dots,(1,n-2)\}}$ in ${\cal A}_{n}$}
\label{proof-convergentdissection}
In this section we consider a Stokes polytope obtained by considering the reference dissection, all whose elements originate from vertex $1$, i.e. $Q=\{14,16,18,\dots,(1,(n-2))\}$. We will show that one can always embed the resulting Stokes polytope inside $(n-3)$ dimensional associahedron ${\cal A}_{n}$.A consequence of such embedding is that a boundary of Stokes polytope obtained by setting $X_{ij} = 0$ (for some $(i,j)$) lie on the corresponding boundary of the associahedron. 

From \cite{1906ppp} and as detailed in Section \ref{convexrealisationstokesposet} we know that the embedding of Stokes polytope inside kinematic space is obtained via sets of constraints classified according to the length of internal paths which ranges from $k=0$ to $k=\frac{n-6}{2}$.  The constraint corresponding to a length $k$ path is :
\begin{align}
\label{lengthkconstraintstokes}
X_{i,i+3+2k}+X_{i+2,i+5+2k}-X_{i,i+5+2k}-X_{i+2,i+3+2k}=d_{i,i+k}\,,
\end{align}
where $i$ takes all odd values in the range $[1,\ldots, n-5-2k]$ and $d_{i,i+k}$ are constants.

We will now show that the constraints \eqref{lengthkconstraintstokes} are contained in the ones associated to the ABHY associahedron corresponding to $T=\{13,14,\dots,1(n-1)\}$, namely $\mathcal A_n$. Notice that $T\subset Q$. The ABHY constraints in this case are given by $s_{ij}=-c_{ij}\ \vert 1 \leq i < j-1\leq n-2$. 
%
%
Using these it is easy to show that :
\begin{align}
\label{momentumconsassociahedron1416}
X_{i,i+3+2k}+X_{i+2,i+5+2k}-X_{i,i+5+2k}-X_{i+2,i+3+2k}=\sum_{I=i}^{i+1}\sum_{J=i+3+2k}^{i+4+2k}\ c_{IJ}\,.
\end{align}
We compare \eqref{lengthkconstraintstokes} and \eqref{momentumconsassociahedron1416} and see that constraints that localize the Stokes polytope ${\cal S}_{n}^{\{14,16,\dots,1(n-2)\}}$ inside kinematic space are realised as a consequence of ABHY constraints. We will now show that ${\cal S}_{n}^{\{14,16,\dots,1(n-2)\}}$ in fact admits a \emph{proper} embedding inside ${\cal A}_{n}$, i.e. on the support of ${\cal S}_{n}^{\{14,16,\dots,1(n-2)\}}$ inside ${\cal A}_{n}$, every $X_{ij}$ variable which does not occur as a vertex of the Stokes polytope remains positive.

Consider the  following additional constraints :
\begin{equation}
\label{addcon}
X_{1j}=d_{1j}\quad\textrm{for all odd $j$.}
\end{equation}
These constraints can be used to locate the Stokes polytope inside kinematic associahedron ${\cal A}_{n}$ as we show below.
\begin{claim}
\label{embeddingclaim1}
For all odd $j$, if $\sum_{I=1}^{j-3}\sum_{J=j+1}^{n-1}\ c_{IJ}\ \leq c_{j-2,j}$ then $\exists\ d_{1j}$ for which $X_{j-2,j}$ and $X_{j-1,j+1}$ are positive.
\end{claim}{}
\begin{proof}
Let us first consider $X_{j-2,j}$ for odd $j$. From the associahedron constraints we have :
%
\begin{align}
X_{j-2,j}=d_{1j}-X_{1,j-1}+\sum_{I=1}^{j-3}\ c_{I,j-1}\,.
\end{align}
Hence $X_{j-2,j}\ \geq 0$ iff 
\begin{equation}\label{1j-1max}
d_{1j}\ +\ \sum_{I=1}^{j-3}\ c_{I,j-1}\ \geq X_{1,j-1}^{\textrm{max}}\,.
\end{equation}
%
To obtain $X_{1,j-1}^{\textrm{max}}$, note that $(p_{1}+\dots+p_{j-2}+p_{j-1}+\ldots+p_{n-1})^{2}=0$, from which we have :
\begin{align}
X_{1,j-1}+X_{j-2,n}=\sum_{I=1}^{j-3}\sum_{J=j-1}^{n-1}c_{IJ}\,.
\end{align}
Since $X_{j-2,n}$ for odd $j$ can be zero, we obtain $X_{1,j-1}^{\textrm{max}}=\sum_{I=1}^{j-3}\sum_{J=j-1}^{n-1}c_{IJ}$.
Thus from equation \eqref{1j-1max} we have,
\begin{align}
\label{d1jlower}
d_{1j}&\geq\sum_{I=1}^{j-3}\sum_{J=j}^{n-1}\ c_{IJ}\,.
\end{align}
Let us now consider the positivity of $X_{j-1,j+1}$ for odd $j$. As before from associahedron constraints we have :
%
%
\begin{align}
X_{j-1,j+1}=X_{1,j+1}-d_{1,j}+\sum_{I=1}^{j-2}c_{Ij}\,.
\end{align}
Since $X_{1,j+1}$ can be zero, $X_{j-1,j+1}\geq 0$ implies : 
\begin{equation}\label{d1jupp}
\ d_{1j}\ \leq \sum_{I=1}^{j-2}\ c_{Ij}\,.
\end{equation}
From equations. \eqref{d1jlower} and \eqref{d1jupp}, we require $d_{1,j}$ to satisfy the following bound :
\begin{equation}\label{boundond}
\sum_{I=1}^{j-3}\sum_{J=j}^{n-1} c_{IJ}\ \leq d_{1j}\ \leq \sum_{I=1}^{j-2}\ c_{Ij}\,.
\end{equation}
We see from \eqref{boundond} that ${\cal S}^{\{14,\dots,1(n-2)\}}_n$ can be embedded inside $\mathcal A_{n}$ only if $c_{ij}$ satisfy :
\begin{equation}
\sum_{I=1}^{j-3}\sum_{J=j+1}^{n-1}\ c_{IJ}\ \leq c_{j-2,j}\,.
\end{equation}
This proves the claim. 
\end{proof}{}
%
%
\begin{claim}
With the above bounds on $d_{1k}$ for all odd $k$, we will now prove that for $(j - i)\ >\ 2$, $X_{ij}\ \geq 0$ when either both $(i,j)$ are even, or both are odd or $i$ is even and $j$ is odd.
\end{claim}{}
\begin{proof}
From the associahedron constraints we have,
\begin{equation}
X_{ij}=X_{1j}-X_{1,i+1}+\sum_{I=1}^{i-1}\sum_{J=i+1}^{j-1}c_{IJ}\,.
\end{equation}
Let us first consider the case when both $i$ and $j$ are even. Then,
\begin{equation}
X_{ij}\ =\ X_{1j}\ -\ d_{1,i+1}\ +\ \sum_{I=1}^{i-1}\sum_{J=i+1}^{j-1}c_{IJ}\,.
\end{equation}
Since $d_{1,i+1}\ \leq  \sum_{I=1}^{i-1}\,c_{I,i+1}$ and as $X_{1j}\ \geq\ 0$, we have $X_{ij}\ \geq 0$ when $i$ and $j$ are both even. A similar analysis shows that $X_{ij}\ \geq 0$ when $i$ and $j$ are both odd. We now consider the final case when $i$ is even and $j$ is odd. 
In this case we have,
\begin{equation}
X_{ij}\ =\ d_{1j}\ -\ d_{1,i+1}\ +\ \sum_{I=1}^{i-1}\sum_{J=i+1}^{j-1}c_{IJ}\,.
\end{equation}
Due to the bounds on $d_{1j}$, we have that $d_{1j}-d_{1i+1}\geq 0$ for $ (j - i)\geq 3$. Hence $X_{ij}$ remain positive for even $i$ and odd $j$ and $ (j - i)\ \geq\ 3$. This proves the claim.
\end{proof}{}
Thus we see that there is always a choice of $d_{1j}$ for which embedding of Stokes polytope ${\cal S}^{\{14,16,\dots,1(n-2)\}}_n$ inside the ABHY associahedron is proper. 

A few remarks are in order.

\begin{itemize}
\item The two dimensional Stokes polytope corresponding to reference quadrangulation $Q=\{14,16\}$ can be convexly realised as a pentagon embedded inside the five dimensional associahedron ${\cal A}_{8}$ via the constraint $X_{15}=d_{15}$ where $d_{15}$ is bounded as $\sum_{I=1}^{2}\sum_{J=5}^{7}\ c_{IJ}\ \leq d_{15}\ \leq \sum_{I=1}^{3}\ c_{I5}$. As stated in Claim \ref{embeddingclaim1}, this will be true only if the $c_{ij}$ themselves satisfy the corresponding bounds. 

\item We note that there is no canonical embedding of ${\cal S}_{n}^{\{14,16,\dots,1(n-2)\}}$ inside ${\cal A}_{n}$. We chose a particular embedding defined by additional constraints in equation \eqref{addcon}. However we could have also considered constraints generated by linear combinations $\sum_{j\ \textrm{odd}}\ a_{1j}X_{1j}\ =\ \textrm{constant}$.    

\item This analysis can be easily generalised to ${\cal S}_{n}^{Q}$ for arbitrary $Q\neq\{14,5n,\dots,(\frac{n}{2}+1,\frac{n}{2}+4)\}$ and it can be shown that ${\cal S}_{n}^{Q}$ can be convexly realised in ${\cal A}_{n}^{T}$ for any $T$ that contains $Q$. 

\item In the next section we consider ${\cal S}_{n}^{\{14,5n,\dots,(\frac{n}{2}+1,\frac{n}{2}+4)\}}$ which arises from reference quadrangulation that consists of non-intersecting diagonals. In contrast to Stokes polytopes associated to other class of quadrangulations, in this case we show that ${\cal S}_{n}^{Q}$ do have \emph{canonical} embedding as a  (hyper-cube) boundary of ${\cal A}_{n}$. 
\end{itemize} 


\subsection{A Canonical Embedding of ${\cal S}_{n}^{\{14,5n,\dots,(\frac{n}{2}+1,\frac{n}{2}+4)\}}$ inside ABHY associahedron}\label{pardiss}
We know from \cite{1906ppp} and the algorithm given in Section.~\ref{convexrealisationstokesposet} that the convex realisation of the Stokes polytope ${\cal S}_{n}^{\{14,5n,\dots,(\frac{n}{2}+1,\frac{n}{2}+4)\}}$ inside kinematic space is given by the constraints :
\begin{align}
\label{145nembeddingequations}
X_{14}\ +\ X_{3n}&=\ d_{1}\cr
X_{ab}\ +X_{a-1,b-1}&=d_{k+2}\,,
\end{align}
for $a=5+k$, $b=n-k$, where $k\in[0,1,\ldots,\frac{n}{2}-4]$.
%
%
\begin{claim}
\label{claim145nembedding}
For $d_{1}=\sum_{I=1}^{2}\sum_{J=4}^{n-1}c_{IJ}$ and $d_{k+2}=c_{4+k,n-k-1}$ for $k\in[0,1,\ldots,\frac{n}{2}-4]$, ${\cal S}_{n}^{\{14,5n,\dots,(\frac{n}{2}+1,\frac{n}{2}+4)\}}$ can be realised as a face of $\mathcal A_{n}$ by setting $X_{13}=0$ and $X_{4+m,n-m}=0$ for $m\in [0,1,\ldots,\frac{n}{2}-3]$. 
\end{claim}
\begin{proof}
Using the associahedron constraints it is easy to see that :
\begin{align}
\label{momentumconsassociahedron145m}
X_{14}\ +\ X_{3n}&=\ \sum_{I=1}^{2}\sum_{J=4}^{n-1} c_{IJ}\cr
%
%
X_{ab}+X_{a-1,b-1}&=X_{a-1,b}+X_{a,b-1}+c_{a-1,b-1}\,,
\end{align}
for $a=5+k$, $b=n-k$, where $k\in[0,1,\ldots,\frac{n}{2}-4]$. We compare these relations with the Stokes polytope constraints in \eqref{145nembeddingequations} and see that when $d_i$ take the values stated in the claim the Stokes polytope is convex realised as the face of $\mathcal A_{n}$ obtained by setting $X_{a-1,b}$ and $X_{a,b-1}$ for $a=5+k$, $b=n-k$, where $k\in[0,1,\ldots,\frac{n}{2}-4]$ to zero. 

As $X_{13}$ never crosses any of the dissections which occur as vertices of the Stokes polytope we can choose $X_{13}=0$. This proves the claim.
\end{proof}

We refer to this embedding as canonical as it is a unique mapping which realises Stokes polytope as a codimension $\frac{n-2}{2}$ face of ${\cal A}_{n}$. 
Such a canonical realisation is not possible for any other reference quadrangulation. This follows from the fact that it is only for $Q\ =\ \{14,5n,\dots,(\frac{n}{2}+1,\frac{n}{2}+4)\}$ that all the vertices of the Stokes polytope correspond to vertices of an associahedron labelled by $Q\ \cup\ \{4n,5(n-1),\ \dots,\ (\frac{n}{2}+1,\frac{n}{2}+3)\}$ with all the co-ordinates except those in $Q$ set to zero.  For a generic quadrangulation whether any embedding is singled out by additional requirements is an interesting question in its own right. We touch upon this issue in Section (\ref{pentinws})

\section{Towards world-sheet forms for quartic interactions}\label{wsforquar} 
In this section we analyse the relationship between $\phi^{4}$ planar amplitudes in terms of positive geometry and projective forms and forms/integrands on the CHY moduli space. For $\phi^{3}$ interactions, it was shown beautifully in \cite{Arkani-Hamed:2017mur, He:2018pue} that the canonical form on the kinematic associahedron is a pushforward of the CHY $\phi^{3}$ ``half-integrand" (also known as the Parke-Taylor form) via the CHY scattering equations. We would like to see if a similar understanding exists between the scattering form for quartic interactions and certain forms on the moduli space. Our idea is to explore planar amplitudes for quartic interactions in terms of pushforward of $d\ln$ forms on \emph{a} moduli space. 

Recall that we have two distinct perspectives on the positive geometry picture in kinematic space. In this section we consider the first one, described in Section.~\ref{phi4fl}, in which one does not consider the Stokes polytope in kinematic space, and all projective forms are defined on the family of ABHY associahedra $\{{\cal A}_{n}^{T}\}$. This encourages us to define for a given quadrangulation $Q$, a world-sheet scattering form $\widetilde{\Omega}^{Q}_{n}$ on ${\cal M}_{0,n}({\bf R})$ and analyze its pushforward on ${\cal A}_{n}^{T}\ \textrm{with}\ Q\ \subset\ T$. As we illustrate below, while the meromorphic piece of the pushforward produces precisely the form proportional to $m^{Q}_{n}$, there is an additional exact $\frac{n-4}{2}$ $d\log$ form with no singularities in the interior of, or on the boundary of ${\cal A}_{n}^{T}$. Thus when viewed as a pushforward from certain equivalence class (defined in equation \eqref{equiv1}) of $d\log$ forms on the CHY moduli space $\overline{{\cal M}}_{0,n}({\bf R})$ onto the kinematic space associahedra via scattering equations, we  recover a form precisely proportional to $m^{Q}_{n}$.

We emphasise that our approach to deriving world-sheet formulae for planar quartic amplitudes is rather simple-minded, as a sophisticated and complete answer to this question presumably requires the existence of a (real) moduli space whose polytopal realisation is the Stokes polytope \cite{devadoss}. However, no such moduli space is known to date \footnote{An unsurprising fact if we think of field theory amplitudes as limits of string amplitudes which do not have quartic vertices. See however \cite{Kalyanapuram:2019nnf}.}\footnote{We are indebted to Satyan Devadoss and Sushmita Venugopalan for several discussions on this issue.} . An extremely interesting take on relating  canonical forms on Stokes polytopes to world-sheet was advocated recently in \cite{Kalyanapuram:2019nnf}. We believe that our analysis is consistent with the results of that paper. 


We now summarise the seminal ideas of \cite{Arkani-Hamed:2017mur, He:2018pue} on which our analysis is based. The planar kinematic form associated to bi-adjoint scalar $\phi^{3}$ theory can be used to define a form on the compactified real moduli space of the punctured Riemann sphere $\overline{{\cal M}}_{0,n}({\bf R})$ \cite{Arkani-Hamed:2017mur}.  On ${\cal M}_{0,n}({\bf R})$ this form turns out to be the Parke-Taylor form. One of the beautiful results established in \cite{Arkani-Hamed:2017mur,He:2018pue} was a derivation of the CHY formula for bi-adjoint scalar $\phi^{3}$ theory by interpreting the canonical form on ${\cal A}_{n}$ as a pushforward of the Parke-Taylor form via scattering equations. Our idea is to use this construction for planar $\phi^{4}$ interactions. We now outline our strategy :
\begin{itemize}
 \item Inspired by \cite{Arkani-Hamed:2017mur}, given a projective, planar scattering form $\Omega_{n}^{Q}$ associated to quadrangulation $Q$, we define the corresponding form on the real section of the CHY moduli space, i.e. ${\cal M}_{0,n}({\bf R})$  by ``formally substituting" $X_{ij}\rightarrow u_{ij}$.  This yields an $\frac{n-4}{2}$ form $\frac{1}{\cal N}\,\widetilde\Omega^{Q}_{n}$ on the moduli space, which has $\log$ singularities on boundaries corresponding to $u_{ij}\rightarrow 0$ with $(ij)$ compatible with $Q$, as expected of quartic interactions. ${\cal N}^{-1}$ is an overall normalisation and is defined to be equal to the number of solutions to scattering equations. Although it appears rather ad-hoc at this stage, it should only be considered as part of the definition and its role will become clear below \footnote{In \cite{Arkani-Hamed:2017mur}, the world-sheet form was obtained by replacing $X_{ij}\rightarrow\sigma_{ij-1}$ and this was precisely the Parke-Taylor $n-3$ form. We will come back to this in Section \ref{scrutiny}.}. 
\item We then study the pushforward of these world-sheet forms via scattering equations in Section \ref{pfa2a} . In light of arguments put forward in Section \ref{phi4fl} we expect the pushforward to produce a form which contains $m^{Q}_{n}$. This indeed turns out to be the case, as we show in the simplest non-trivial example of $n=6$ that the pushforward is the sum of such a form and an additional \emph{exact} form ! The exact form has no singularities in the interior of, or on the boundaries of the embedding associahedron. We thus define a world-sheet form $\widehat{\Omega}^{Q}_{6}$ whose pushforward indeed produces the $n$ point amplitude compatible with $Q$. $\widehat{\Omega}^{Q}_{6}$ parametrises an equivalence class of $d\ln$ forms $[\widetilde{\Omega}^{Q}_{6}]$ where equivalence is defined by addition of analytic exact forms \footnote{As we are working with real moduli spaces whose compactification is an associahedron as opposed to the ones discussed in \cite{cohomologyannals}, we do not identify these equivalence classes with any cohomology groups. We discuss this point in more detail below.}. We then write down the most general pushforward formula for any $\widehat{\Omega}^{Q}_{n}$ using the  diffeomorphism between $\overline{{\cal M}}_{0,n}({\bf R})$ and any of the ABHY associahedra.
\end{itemize}

\subsection{Towards world-sheet forms}
\label{Towards world-sheet forms}
%
Following the discussion above, we first consider the simplest $n=6$ case with $Q=\{14\}$ being the reference quadrangulation. Using $\Omega^{\{14\}}_{6}=d\ln X_{14}-d\ln X_{36}$ \footnote{Note that, when restricted to the kinematic space associahedron this yields a 1-form which is the canonical form on $ST^{\{14\}}_{6}$.}, we define a 1-form on ${\cal M}_{0,6}({\bf R})$ as, 
%
\begin{align}\label{14ws}
\widetilde\Omega_{6}^{\{14\}}&:=\frac{1}{{\cal N}}\left(d\ln u_{14}-d\ln u_{36}\right)\cr
&=\frac{1}{{\cal N}}\left(\frac{d\sigma_{3}}{\sigma_{3}(1-\sigma_{3})}-\frac{d\sigma_{4}}{\sigma_{4}}\right)\,,
\end{align}
where in the second line we performed the usual gauge fixing $\sigma_{1}=0, \sigma_{5}=1,\sigma_{6}=\infty$. We would now like to see the push-forward of this form on ${\cal A}_{6}$. As neither of the two channels $X_{14}$ or  $X_{36}$ crosses $X_{13}$, we can set $X_{13}=0$ without loss of generality. This implies that on the world-sheet we can set $\sigma_{2}=0$.

The CHY scattering equations can be used to show that, 
\begin{align}
\sigma_{3}=\frac{X_{14}\sigma_{4}}{c_{14}\ +\ c_{24}\ +\ \sigma_{4}(c_{15} + c_{25})}\,,\quad
\sigma_{4}=\frac{X_{15}(1-\sigma_{3})\ +\ \sigma_{3}c_{35}}{[\sum_{I=1}^{3}c_{I5}\ -\ \sigma_{3}\sum_{I=1}^{2}c_{I5}]}\,.
\end{align}
These equations are complicated to solve and hence we consider a specific kinematic configuration where $c_{15}+c_{25}=0$ \footnote{This choice of kinematics is a degenerate case that corresponds to certain collinear limits. We use it only for the purpose of illustration as will become clear below.}. In this case, we obtain for $\sigma_3$ and $\sigma_4$ the following :
\begin{align}
\sigma_3=\frac{X_{14}\,X_{15}}{c_{35}\sum_{i=1}^2c_{i4}-c_{35}\,X_{14}+X_{14}\,X_{15}}\,,\quad
\sigma_4=\frac{X_{15}\sum_{i=1}^2c_{i4}}{c_{35}\sum_{i=1}^2c_{i4}-c_{35}\,X_{14}+X_{14}\,X_{15}}\,.
\end{align}
Hence for such a restricted choice of kinematics, there is a unique real solution to the scattering equations and thus ${\cal N}=1$.

We now use these equations to push-forward $\widetilde{\Omega}^{(14)}_{6}$ onto ${\cal A}_{6}$. After some algebra and using the fact that for this special choice of kinematics, $X_{14}+X_{36}=c_{14}+c_{24}$ we see that the push-forward is given by,
\begin{align}
\textrm{push-forward}=dX_{14}\left(\frac{1}{X_{14}}+\frac{1}{X_{36}}\right)+\textrm{analytic term}\,.
\end{align}
For our choice of kinematics, the analytic term is given by :
\begin{equation}\label{anal-nui}
\textrm{Analytic term}=d\ln\left[c_{35}\sum_{i=1}^{2}c_{i4}-c_{35}X_{14}+X_{14}X_{15}\right]\,.
\end{equation}
This term can be understood as follows. For our choice of kinematics, it can be verified that
\begin{align}
X_{14}&=\sum_{i=1}^{2}\ c_{i4}u_{14}=:f_{14}u_{14}\cr
X_{36}&=\sum_{i=1}^{2}c_{i4}u_{35}u_{36}=: f_{36}u_{36}\,.
\end{align}
We thus see that under push-forward by scattering equations,
\begin{equation}
\widetilde{\Omega}^{\{14\}}_{6}\rightarrow\left(\frac{1}{X_{14}}+\frac{1}{X_{36}}\right)dX_{14}-d\ln\left[\frac{f_{14}}{f_{36}}\right]\,.
\end{equation}
This equation can be re-written as 
\begin{align}
\widetilde{\Omega}^{\{14\}}_{6}+d\ln\left[\frac{f_{14}}{f_{36}}\right]\ \rightarrow\left(\frac{1}{X_{14}}+\frac{1}{X_{36}}\right)\ dX_{14}\cr
\implies d\ln\{\frac{u_{14}\ f_{14}}{u_{36}\ f_{36}}\}\rightarrow\left(\frac{1}{X_{14}}+\frac{1}{X_{36}}\right)\ dX_{14}
\end{align}
where $f_{14}=c_{14}+c_{24}$ and $f_{36}=\frac{\sigma_3-\sigma_4}{\sigma_3-1}\frac{1}{\sigma_4}\,(c_{14}+c_{24})$.
Several comments are in order.
\begin{itemize}
\item In light of our analysis in Section (\ref{phi4fl}), the result of this simple example is not surprising. There it was shown that the restriction of the planar scattering form $\Omega^{Q}_{n}$ onto an ABHY Associahedron ${\cal A}_{n}^{T}$ (for any triangulation $T$ which contains $Q$) is proportional to $m_{n}^{Q}$. As ${\cal M}_{0,n}({\bf R})$ is diffeomorphic to ${\cal A}_{n}^{T}$ via scattering equations, it is rather expected that a world-sheet form with precisely the same singularity structure as $\Omega^{Q}_{n}$ produces the same singular contributions as $\Omega^{Q}_{n}\vert_{{\cal A}_{n}^{T}}$ under push-forward by scattering equations. The above example illustrates this explicitly. A beautiful geometric manifestation of this observation arises via intersection theory and has been explored in \cite{Kalyanapuram:2019nnf}. Due to the argument presented above, we expect  that this result can be generalised to any $n$. 
\item In contrast to the Parke-Taylor form, $\widetilde{\Omega}^{(14)}_{6}$ also produces a term which is analytic in ${\cal A}_{6}$         (it is in fact analytic in the positive quadrant of the planar kinematic space). This makes their push-forwards more intricate than those for the canonical top forms where such maps were just a manifestation of the CHY formula for (planar) $\phi^{3}$ interactions. 
\item The analytic piece vanishes when $X_{46}=0$. In $\overline{{\cal M}}_{0,n}({\bf R})$ this simply corresponds to going to the $\sigma_{4}\rightarrow 1$ boundary and hence reduces the world-sheet form to $\frac{d\sigma_{3}}{\sigma_{3}(1-\sigma_{3})}$.
\item  The $X_{46}=0$ constraint corresponds to the convex realisation of $ST^{\{14\}}$ in ${\cal A}_{6}$, and as we show below, has a nice counterpart on the world-sheet. Namely, the image of the (convexly realised) $ST^{\{14,5n,\dots,1(n-2)\}}$ is a  boundary of ${\cal M}_{0,n}({\bf R})$ and the restriction of $\widetilde{\Omega}_{n}^{\{14,5n,\dots,1(n-2)\}}$ onto this boundary is such that its push-forward via the CHY scattering equations produces precisely the amplitude. This helps us in writing a CHY type formula for these simplest kind of Stokes polytopes. 
\end{itemize}
\subsection{Lower forms on ${\cal M}_{0,n}({\bf R})$}\label{pfa2a}
%
%
%
In this section we analyse the pushforward of $\widetilde{\Omega}^{Q}_{n}$ via the CHY scattering equations. Generalizing the result in \cite{Arkani-Hamed:2017mur}, in Appendix \ref{messscatt} we propose a map from world-sheet associahedron to ABHY associahedron with arbitrary reference triangulation. Inspired by this map and the discussion in Section \ref{phi4fl} we make the following claim.
\begin{claim}
The pushforward of $\widetilde{\Omega}^{Q}_{n}$ via CHY scattering equations equals $m_{n}^{Q}\ \prod_{(ij)\in Q}dX_{ij}$ up to an exact form.
\end{claim}
\noindent {\bf Proof} : 
Our primary argument is the following : $\widetilde{\Omega}^{Q}_{n}$ is a $d\log$ form with simple poles on a boundary subset of $\overline{{\cal M}}_{0,n}({\bf R})$, where the subset of boundaries is in one to one correspondence with $Q$-compatible quadrangulations. Scattering equations are diffeomorphisms between the world-sheet associahedron and any of the ABHY associahedra \footnote{Strictly speaking there is no proof of this conjecture to the best of our knowledge. However, it has been verified up to $n=10$ \cite{Arkani-Hamed:2017mur}.}.  Hence if we consider an equivalence class of the $d\ln$ forms
\begin{equation}\label{equiv1}
[\widetilde{\Omega}^{Q}_{n}]\ =\ \{\widetilde{\Omega}^{Q}_{n}\ +\ \textrm{exact-form}\}
\end{equation} 
where by ``exact-form" we mean forms which are analytic on $\overline{{\cal M}}_{0,n}({\bf R})$, the scattering equations induce the map at the level of $[\widetilde{\Omega}^{Q}_{n}]$. \emph{As the restriction of planar scattering form $\Omega^{Q}_{n}$ (which is also a closed form) to ${\cal A}_{n}^{T}$ equals $m_{n}^{Q}\prod_{(ij)\in Q}dX_{ij}$}, pushforward of $\widetilde{\Omega}^{Q}_{n}$ to ${\cal A}_{n}^{T}$ produces $m_{n}^{Q}\prod_{(ij)\in Q}dX_{ij}$ up to analytic (non-singular) form. Hence CHY scattering equations map $\widetilde{\Omega}^{Q}_{n}$ to $\Omega^{Q}_{n}\vert_{{\cal A}_{n}^{T}}$ up to an exact form. 

We now quantify this argument in  detail. 
Consider the map between the world-sheet associahedron and kinematic space associahedron as expressed in Appendix \ref{messscatt} :
\begin{align}\label{worldABHYmaptext}
     X_{ij} = \sum_{\substack{ i\leq k \leq j-2\\ j \leq \ell \leq i-2 }} c_{k\ell} \frac{\sigma_{i-1,\ell }\sigma_{k,j-1}}{\sigma_{k,\ell} \sigma_{i-1,j-1}} = \sum_{\substack{ i\leq k \leq j-2\\ j \leq \ell \leq i-2 }} c_{k\ell} \prod_{\substack{i \leq  p \leq k \\ j\leq q \leq \ell} } u_{pq}\,. 
\end{align}
For simplicity we restrict our attention to ${\cal A}_{n}$ (as opposed to ${\cal A}_{n}^{T}$) in the kinematic space. The  associahedron constraints associated to ${\cal A}_{n}$ are :
\begin{align}\label{assopent}
X_{ij} = X_{1j} - X_{1,i+1} + \sum_{\substack{1\leq k \leq i-1\\ i+1 \leq l \leq j-1}} c_{k\ell}.
\end{align}
Combining equations \eqref{worldABHYmaptext} and \eqref{assopent}, and after some algebra we get :
\begin{align}\label{wtf1}
X_{ij} = \sum_{\substack{1\leq k \leq i-1\\ i+1 \leq \ell \leq j-1}} c_{k\ell}  \prod_{\substack{k+1 \leq  p \leq i \\ \ell+1\leq q \leq n} } u_{pq} + \sum_{\substack{ i\leq k \leq j-2\\ j \leq \ell \leq n-1 }} c_{k\ell} \prod_{\substack{1 \leq  p \leq k \\ j\leq q \leq \ell} } u_{pq} +  \sum_{\substack{ 1\leq k \leq i-1\\ j \leq \ell \leq n-1 }} c_{k\ell}   \prod_{\substack{1 \leq  p \leq k \\ j\leq q \leq\ell} } u_{pq}  \prod_{\substack{k+1 \leq  r \leq i \\ j\leq s \leq n} } u_{r,s}\,.
\end{align}
Notice that each $X_{ij}$ is of the form 
\begin{equation}\label{xuf}
X_{ij}= u_{ij}f_{ij}(u),
\end{equation}
 where assuming the existence of diffeomorphism between world-sheet and kinematic space associahedra, we deduce that $f_{ij}$ is a non-vanishing polynomial in the other $u_{k\ell}$. 

Let us now consider the following form which equals $\widetilde{\Omega}^{Q}_{n}$ up to a $d\log$ form :
\begin{equation}\label{CHYform}
    \widehat{\Omega}^{Q}_{n} := \frac{1}{{\cal N}}\ \sum_{Q' \in ST^Q} \text{sgn}(Q') \bigwedge_{(ij)\in Q'} \left(d\log[u_{ij}] + d\log[f_{ij}(u)]\right)\,.
\end{equation}
Using equation \eqref{xuf} it can be readily checked that, \emph{on each solution to the scattering equation} the pushforward of this form produces $\Omega^{Q}_{n}\vert_{{\cal A}_{n}^{T}}$. Thus summing over all the solutions and dividing by ${\cal N}$ produces the amplitude. 

The primary argument written above  equation \eqref{worldABHYmaptext} then implies that $\widehat{\Omega}^{Q}\ -\ \widetilde{\Omega}^{Q}$ is exact, i.e. it has no singularities on the CHY moduli space.  We note that as $\widehat{\Omega}^{Q}_{n}\ -\ \widetilde{\Omega}^{Q}_{n}$ is sum of $\frac{n-4}{2}$ forms each of which has at least one differential of the form $d\ln[f_{ij}(u)]$ and hence the primary argument is equivalent to the claim that all the functions $\ln f_{ij}$ are analytic on $\overline{{\cal M}}_{0,n}({\bf R})$.

Let us illustrate this with an example. We again consider the simplest non-trivial case of $n=6$ with $Q=\{14\}$ being the reference quadrangulation. From Section \ref{Towards world-sheet forms} we consider the following 1-form on ${\cal M}_{0,6}({\bf R})$ :
%
\begin{align}\label{14ws}
\widetilde\Omega_{6}^{\{14\}}&=d\ln u_{14}-d\ln u_{36}\,.
\end{align}
%
From scattering equations we have :
%
\begin{align}
\label{se1436}
     X_{14} &= u_{14}f_{14}(u)\,,\cr
      X_{36}&= u_{36}f_{36}(u)\,,\cr
   \text{where}\quad f_{14}(u)&=\left(c_{14}+ c_{15}u_{15} + c_{24}u_{24}+c_{25}u_{15}u_{24}u_{25}\right)\,,\cr
     f_{36}(u)&=  \left(c_{14}u_{25}u_{35}u_{26}+c_{15}u_{26}+c_{24}u_{35}+c_{25} \right)\,.
\end{align}
Therefore $\widehat{\Omega}^{\{14\}}_{6}$ in \eqref{CHYform} is given by :
\begin{align}\label{extfor14}
        \widehat{\Omega}^{\{14\}}_{6}= d\log u_{14} - d\log u_{36} +d\log\left(\frac{c_{14} + c_{15}u_{15} + c_{24}u_{24} + c_{25}u_{15}u_{24}u_{25}}{c_{14}u_{25}u_{35}u_{26} + c_{15}u_{26} + c_{24}u_{35} + c_{25}}\right)\,.
\end{align}
Since $ c_{14} + c_{15} + c_{24} + c_{25} \geq f_{14} \geq c_{14} $ and $ c_{14} + c_{15} + c_{24} + c_{25} \geq f_{36} \geq c_{25}$, the extra term is an exact form. Using \eqref{se1436} we see that pushforward of $\widehat{\Omega}^{\{14\}}_{6}$ equals $\Omega^{\{6\}}\vert_{{\cal A}_{6}}=m_{6}^{\{14\}}\ dX_{14}$.



Although we quantified our arguments for ${\cal A}_{n}$, it can be generalised to any $Q$ and any ${\cal A}_{n}^{T}$ using the analysis in Appendix \ref{messscatt}. Thus for any $Q$ and $n$ we have a pushforward map  at the level of the equivalence classes  $[\widetilde{\Omega}^{Q}_{n}]$  
\begin{align}\label{pfmap}
\boxed{\widehat{\Omega}^{Q}_{n}\ \xrightarrow{\text{scattering equations}}\ m^{Q}_{n}\prod_{(ij)\ \in\ Q}\ dX_{ij}}
\end{align}

It is clear that $\widehat{\Omega}^{Q}_{n}$ are representatives of the equivalence class $[\widetilde{\Omega}^{Q}_{n}]$ defined in equation \eqref{equiv1}. 
In the case of complexified moduli space ${\cal M}_{0,n}$, $[\widetilde{\Omega}^{Q}_{n}]$ belongs to the cohomology groups $H^{\frac{n-4}{2}}({\cal M}_{0,n})$ \footnote{We note that Parke-Taylor forms form a basis for $H^{n-3}({\cal M}_{0,n})$ \cite{Mizera:2019gea}, for both the integral cohomology classes or cohomology over ${\bf C}$. Whether the $Q$-compatible forms generate such a basis for $\frac{n-4}{2}$ dimensional cohomology classes remains outside the scope of the paper.}. But as the compactification of the real moduli space that we are working with is simply an associahedron, we do not identify these equivalence classes with  a cohomology class \footnote{As an aside we note that cohomology classes of a different compactification of the real moduli space (which for $n=4$ is simply the one point compactification) has been analysed in \cite{Devadoss,cohomologyannals}. This compactification leads to $\overline{{\cal M}}_{0,n}({\bf R})$ which are Eilenberg-MacLane spaces $K(\pi, n)$ and are non-orientable for $n\ \geq 5$. It will be extremely interesting to analyse the relationship (if any) of ``Q-compatibility" with rank $\frac{n-4}{2}$ cohomology classes for such compactification.} .

\subsection{Further scrutiny of the world-sheet form}
\label{scrutiny}
In the previous section we argued that the push-forward of the world-sheet form $\widehat{\Omega}^{Q}_{n}$ by scattering equations generate the corresponding amplitude contribution $m_{n}^{Q}$. $\widehat{\Omega}^{Q}_{n}$ was obtained by considering the planar scattering form in kinematic space and simply replacing $X_{ij}$ with $u_{ij}\ f_{ij}$. 

We could try to apply the same idea for $\phi^{3}$ interactions and revisit the results of \cite{Arkani-Hamed:2017mur}. However in this case, we know that the world-sheet form (whose push-forward by scattering equations yields the planar amplitude) is the well known Parke-Taylor form. In the $\sigma_{1}=0,\sigma_{n-1}=1,\sigma_{n}=\infty$, this form is given by,
\begin{equation}
\label{wsheetform}
\omega^{\textrm{ws}}_{n}\ =\ \frac{d\sigma_{2}\wedge\dots\wedge d\sigma_{n-2}}{\sigma_{2}\ \sigma_{23}\ \dots\ \sigma_{n-2,n-1}}\,.
\end{equation}
It is rather obvious that $\widehat{\Omega}_{n}$ is not equal to the Parke-Taylor form.  Although, $\widehat{\Omega}_{n}$  is defined on the world-sheet, it depends on the kinematic data . This is in contrast to Parke-Taylor form which is the canonical top form on the Moduli space. Hence the ideas laid out in previous section seems to break down for cubic theory. If this were the case, it would be rather surprising and would require more scrutiny of the definition of $\widehat{\Omega}_{n}$. We will now argue by means of an example  that this is not the case.  Namely, on the solutions of scattering equations, the two forms are equal and the push-forward of $\widehat{\Omega}_{n}$ and the Parke-Taylor form both yield the canonical form on the kinematic space Associahedron.  We would like to emphasise that this comparison also shows that for quartic interactions $\widehat{\Omega}_{n}$ is not the final answer and there must exist a form ``intrinsic" to the world-sheet whose push forward will produce the corresponding contribution to the amplitude. However derivation of such form has eluded us so far. 
%
%
%
%

Our claim is that even in the case of $\phi^{3}$ theory, we have the following. 
\begin{equation}
\widehat{\Omega}_{n}-\omega^{\textrm{ws}}_{n}\ \rightarrow\ 0
\end{equation}
Even though we lack a general proof of this statement, we verify it for $n\ =\ 5$ point case \footnote{Verification for the first non-trivial, that is $n\ =\ 4$ case is trivial and we leave it as an exercise for the reader.}, i.e. we compute 
\begin{displaymath}
\sum_{\textrm{solns to scatt eqns}}\widehat{\Omega}_{5}-\omega^{\textrm{ws}}_{n}
\end{displaymath}
and show that it indeed vanishes. 

%
%
%
For $n=5$ case, the scattering equations for $X_{13}$ and $X_{14}$ in our choice of gauge  $(\sigma_1=0,\sigma_4=1,\sigma_5=\infty)$ are given by \cite{Arkani-Hamed:2017mur} :
\begin{align}\label{scattereq}
    X_{13} &= \frac{\sigma_{2}}{\sigma_{3}}(c_{13}+ \sigma_{3}c_{14}) \cr
        X_{14} &= \frac{1}{1-\sigma_{2}}((\sigma_{3}-\sigma_{2})c_{24} + \sigma_{3}(1-\sigma_{2})c_{14} ).
\end{align}
As shown in \cite{Arkani-Hamed:2017mur}, push-forward of world-sheet forms on the kinematic space associahedron is obtained by substituting solutions to the scattering equations. For generic choices of $c_{ij}$ computing the push-forward  of $\widehat{\Omega}$ is rather complicated and hence in the interest of pedagogy, we consider two choices of $c_{ij}$. In the first case, we consider $c_{14}=0$ while other $c_{ij}$
s are arbitrary. We see that for this choice of $c_{ij}$ $\widehat\Omega_5$ matches $\omega_{5}^{\textrm{ws}}$ and hence their push forward under the scattering equations trivially match.
%
%

We then consider the case where $c_{ij}=1$ for $\forall\ (i,j)$ and compute the push-forward. The scattering equations have two distinct solutions and hence the overall normalisation factor of $\widehat{\Omega}_{5}$ is $\frac{1}{2}$. As we show explicitly in Appendix \ref{mrunmay1}, the push-forward of $\widehat{\Omega}_{5}$ once again matches the push-forward of Parke-Taylor form. In fact we have checked that the push-forwards match for generic $c_{ij}$, but the intermediate expressions are quite complicated and we don't present them here.

The fact that $\widehat{\Omega}_{5}$ matches with the Parke-Taylor form only on the solutions to the scattering equations has a precedence in CHY formalism. For example, the reduced Pfaffian which generically shows up in CHY integrands is well defined only on the solution to the scattering equations.

We conclude this section with a few comments 

\begin{itemize}

\item In a recent work \cite{Kalyanapuram:2019nnf}, the author initiated the study of intersection theory for Stokes polytopes and derived formulae for $m^{Q}_{n}$ from intersection numbers of certain $d\log$ forms on the (complexified) moduli space ${\cal M}_{0,n}$. We believe that, as in the case of bi-adjoint amplitudes \cite{Mizera:2017rqa}, this approach is closely tied to the pushforward map in equation \eqref{pfmap}. However, a detailed comparison of the two approaches is beyond the scope of this paper.

%

\item In the $n=6$ case, we wrote an explicit form $\widehat{\Omega}^{\{14\}}_{6}$  whose pushforward produces partial amplitude $m_{6}^{\{14\}}$. This form is obtained by subtracting out the exact form $d\ln f[u_{14}]$ from the world-sheet form $\widetilde{\Omega}^{\{14\}}_{6}$. Interestingly enough, it can be readily checked that the exact form $d\ln f[u_{ij}]$  vanishes when $X_{46}=0$. The constraint precisely corresponds to the  convex realisation of ${\cal S}_6^{\{14\}}$ in ${\cal A}_{6}$ defined in Section (\ref{pardiss}).

\item In $\overline{{\cal M}}_{0,6}({\bf R})$ $X_{46}\ =\ 0$ is mapped via the scattering equations to the $\sigma_{4}\rightarrow 1$ boundary. We refer to this boundary as a world-sheet Stokes geometry and denote it as $\widetilde{{\cal S}}^{\{14\}}_6$. Hence restriction of the world-sheet form \eqref{14ws} to $\widetilde{{\cal S}}^{\{14\}}_6$ equals $\frac{d\sigma_{3}}{\sigma_{3}(1-\sigma_{3})}$ such that its pushforward via scattering equations equals $m_{6}^{\{14\}}\ dX_{14}$. 

\end{itemize}

In the following section we expand on these observations. We show that the image of the convexly realised ${\cal S}^{\{14,5n,\dots,1(n-2)\}}_n$ under diffeomorphism generated by CHY scattering equations is a  boundary of ${\cal M}_{0,n}({\bf R})$ and restriction of $\widetilde{\Omega}_{n}^{\{14,5n,\dots,1(n-2)\}}$ on this boundary is such that the corresponding pushforward under CHY scattering equations precisely produces the amplitude thus aiding us in writing a CHY type formula for the Stokes polytopes with hyper-cube topology.

\section{World-sheet Stokes Geometries}\label{wsst1}
For the bi-adjoint scalar amplitudes, the pushforward of the canonical form on world-sheet associahedron can be written in an integral representation which is the CHY formula for bi-adjoint scattering amplitudes. 
In this section we try to analyse if a similar formula can be written down in the present case, i.e. if the pushforward map defined in equation (\ref{pfmap}) can be expressed as a ``CHY-inspired" integral formula \cite{Cachazo:2013iea}.
 
We first consider the quadrangulation consisting of parallel diagonals $\{14,5n,\dots,\frac{n}{2}+1,\frac{n}{2}+4\}$ and analyse the pushforward of world-sheet form from the  perspective mentioned towards the end of previous section. 

We will see that the unique convex realisation of ${\cal S}^{\{14,5n,\ldots,(\frac{n}{2}+1,\frac{n}{2}+4)\}}_{n}$ as boundary of $\mathcal A_n$ helps us in writing ``CHY-type" integral formula for these special class of quadrangulations.

In Section (\ref{wssp14}) we use the canonical embedding of ${\cal S}^{\{14,5n,\ldots,(\frac{n}{2}+1,\frac{n}{2}+4)\}}_{n}$ as a boundary of ${\cal A}_{n}$ to define the world-sheet Stokes geometry $\widetilde{{\cal S}}^{\{14,5n,\ldots,(\frac{n}{2}+1,\frac{n}{2}+4)\}}_{n}$. We then show that restriction of $\Omega^{Q}_{n}$ on $\widetilde{{\cal S}}^{Q}_n$ (which we denote as $\widetilde{\omega}^{Q}_{n}$) is such that its pushforward via scattering equations equals scattering amplitude. We finally write this pushforward map as an integral formula on ${\widetilde{{\cal S}}}^{\{14,5n,\ldots,(\frac{n}{2}+1,\frac{n}{2}+4)\}}_{n}$.

We then argue that there is an ``intrinsic" characterization of $\widetilde{{\cal S}}^{\{14,5n,\dots,(\frac{n}{2}+1,\frac{n}{2}+4)\}}_n$ in $\overline{{\cal M}}_{0,n}({\bf R})$. This characterisation is obtained by defining $\widetilde{{\cal S}}^{Q}_n$ as a $\frac{n-4}{2}$ dimensional positive-geometry embedded in the moduli space on which the exact form $\widehat{\Omega}_{n}^{Q}\ -\ \widetilde{\Omega}^{Q}_{n}$ vanishes. This characterisation leads us to a definition of the world-sheet Stokes geometry $\widetilde{{\cal S}}^{Q}_n$ for other class of quadrangulations in Section \ref{pentinws}.

\subsection{World-sheet Stokes geometries for $Q\ =\ \{14,5n,\dots,(\frac{n}{2}+1,\frac{n}{2}+4)\}$}\label{wssp14}
We begin this section by considering again the $n=6$ case with $Q=\{14\}$ being the reference quadrangulation.
Let us note that the exact form on the RHS of equation \eqref{extfor14} has an interesting structure. If we approach the co-dimension 2 boundary of $\overline{{\cal M}}_{0,6}({\bf R})$ obtained by setting $u_{46}$ and $u_{13}$ to $0$, it vanishes. Thus the restriction of $\widetilde\Omega^{\{14\}}_6$ on to this boundary is such that its pushforward by scattering equations produces exactly the partial amplitude $m^{\{14\}}_6$.
%
%

In kinematic space this boundary corresponds to $X_{46}=0$ and $X_{13}=0$, which from our discussion in Section \ref{pardiss} gives precisely the canonical convex realisation of ${\cal S}_6^{\{14\}}$ as a boundary of $\mathcal{A}_6$. Image of this convex realisation under (inverse of) diffeomorphism generated by scattering equation is thus a one-dimensional positive geometry in $\overline{{\cal M}}_{0,6}({\bf R})$ which we will refer to as world-sheet Stokes Geometry and denote as $\widetilde{{\cal S}}^{\{14\}}_6$. We denote the restriction of $\frac{1}{{\cal N}}\,\widetilde{\Omega}^{\{14\}}_{6}$ on to $\widetilde{{\cal S}}^{\{14\}}_6$ as $\widetilde{\omega}^{\{14\}}_{6}$  and it is given by :
\begin{align}\label{n6wshc}
\widetilde{\omega}^{\{14\}}_{6}=
\frac{d u_{14}}{u_{14}(1-u_{14})}\,.
\end{align}
Let us now generalise the above discussion to arbitrary $n$ when $Q=\{14,5n,\dots,(\frac{n}{2}+1,\frac{n}{2}+4)\}$. As we know from Section \ref{pardiss} and the discussion above,  world-sheet Stokes Stokes geometry $\widetilde{{\cal S}}^{\{14,5n,\dots,(\frac{n}{2}+1,\frac{n}{2}+4)\}}_{n}$ is obtained by setting $u_{ij}=0$ $\forall\ (ij)\in\ Q^{\textrm{c}}:=\{13,4n,5(n-1),\dots,(\frac{n}{2}+1,\frac{n}{2}+3)\}$. The restriction of $\widetilde{\Omega}^{\{14,5n,\dots,(\frac{n}{2}+1,\frac{n}{2}+4)\}}_{n}$ on to this  Stokes geometry is an immediate generalisation of \eqref{n6wshc} :
\begin{align}\label{smallom145n}
\widetilde{\omega}^{\{14,5n,\dots,(\frac{n}{2}+1,\frac{n}{2}+4)\}}_{n}&=\frac{du_{14}\ \wedge\ du_{5n}\wedge\ \dots\ \wedge\ du_{(\frac{n}{2}+1),(\frac{n}{2}+4)}}{u_{14}(1-u_{14})\ \dots\ u_{\frac{n}{2}+1,\frac{n}{2}+4}(1-u_{\frac{n}{2}+1,\frac{n}{2}+4})}\cr
&=\bigwedge_{(ij)\in \{14,5n,\dots,(\frac{n}{2}+1,\frac{n}{2}+4)\}}\,\frac{du_{ij}}{u_{ij}(1-u_{ij})}\,.
\end{align}

We will now verify that under pushforward by scattering equations, \eqref{smallom145n} reproduces the partial amplitude $m_{n}^{\{14,5n,\dots,(\frac{n}{2}+1,\frac{n}{2}+4)\}}$. We notice the following linear relations between $X_{ij}$ and the corresponding $u_{ij}$ (see Appendix \ref{linear} for a derivation of these relations) :
\begin{align}
\label{linxu}
X_{14}&=u_{14}\sum_{4\le j\le n-1}\,(c_{1j}+c_{2j})\,,\cr
X_{ab}&=u_{ab}\ c_{a-1,b-1},\quad a=5+k,\ b=n-k,\quad k\in [0,\ldots,\frac{n}{2}-4]\,,\cr
X_{3n}&=u_{3n}\sum_{4\le j\le n-1}\,(c_{1j}+c_{2j})\,,\cr
X_{ab}&=u_{ab}\ c_{ab},\quad a=4+k, b=n-k-1,\quad k\in[0,1,\ldots \frac{n}{2}-4]\,.
\end{align}

As the canonical form on $\widetilde{{\cal S}}_{n}^{\{14,5n,\dots,(\frac{n}{2}+1,\frac{n}{2}+3)\}}$ is simply a direct product of $d \ln$ 1-forms \eqref{smallom145n}, we can immediately see that its pushforward using scattering equations \eqref{linxu} is the canonical form on ${\cal S}^{\{14,5n,\dots,(\frac{n}{2}+1,\frac{n}{2}+4)\}}_{n}$:
\begin{align}
\frac{1}{\vert\textrm{solns}\vert}\ \sum_{\textrm{solns}}\prod_{(ij)\in Q}\frac{d u_{ij}}{u_{ij}(1-u_{ij})}=m^{Q}_n\,\wedge_{(ij)\,\in\,Q}\,dX_{ij}\,,
\end{align} 
where the sum is over all the solutions of \eqref{linxu} and $Q=\{14,5n,\dots,(\frac{n}{2}+1,\frac{n}{2}+4)\}$. Thus for this class of quadrangulations, the pushforward can be expressed as a ``CHY-type" integral formula where the integrand is an $\frac{n-4}{2}$ form and integration is over the world-sheet Stokes geometry :
\begin{equation}\label{hypecubechy}
\boxed{m_{n}^{Q}\ =\ \frac{1}{\vert\textrm{solns}\vert}\ \int\ \prod_{(ij)\in\ Q}\frac{du_{ij}}{u_{ij}(1-u_{ij})}\delta (X_{ij}-c_{ij}u_{ij})}
\end{equation}
where $Q=\{14,5n,\dots,(\frac{n}{2}+1,\frac{n}{2}+4)\}$.

The definition of $\widetilde{{\cal S}}^{\{14,5n,\dots,(\frac{n}{2}+1,\frac{n}{2}+4)\}}_{n}$ above is as an image of ${\cal S}^{\{14,5n,\dots,(\frac{n}{2}+1,\frac{n}{2}+4)\}}_{n}$ under scattering equations. However it is easy to see that this world-sheet Stokes geometry can also be defined by evaluating the exact form $\widehat{\Omega}\ -\ \widetilde{\Omega}$ along the lines of equation \eqref{extfor14} and locating a co-dimension $\frac{n-2}{2}$ hyper-surface in $\overline{{\cal M}}_{0,n}({\bf R})$ on which this form vanishes. As the exact form is a linear combination of world-sheet forms each of which contains at least one $d\ln f_{ij}[u]$, the zero-locus hyper surface is obtained by setting $f_{ij}=\textrm{constant}.$ 

In the present case when $Q=\{14,5n,\dots,(\frac{n}{2}+1,\frac{n}{2}+4)\}$ it can be readily verified that each of the $f_{ij}$ is given by, 
\begin{equation}\label{extfforn}
f_{ij} = c_{i,j} +   \sum_{\substack{i \leq k \leq j-2 \\ j \leq \ell \leq i-2 \\ (k,\ell) \neq (i,j) }} c_{k,\ell} \prod_{ \substack{i\leq p \leq k \\ j \leq q \leq \ell \\ (p,q) \neq (i,j)} } u_{p,q}\,.
\end{equation}
A simple way to see this is that for this reference quadrangulation, $\widetilde{\Omega}^{Q}_{n}$ is a product of $1$-forms for each hexagon which is such that in each such hexagon precisely one of the $(ij)\ \in\ Q$ is a reference quadrangulation. It can be shown that all of $f_{ij}=\textrm{constant}$ is equivalent to $u_{mn}=0\,\,\forall\,\,(mn)\,\in\,\{13,4n,5(n-1),\dots,(\frac{n}{2}+1,\frac{n}{2}+3)\}$. 

Hence the upshot is that $\widetilde{{\cal S}}^{\{14,5n,\dots,(\frac{n}{2}+1,\frac{n}{2}+4)\}}_n$ is a co-dimension $\frac{n-2}{2}$ boundary of $\overline{{\cal M}}_{0,n}({\bf R})$ and $\widetilde{\omega}^{\{14,\dots,(\frac{n}{2}+1,\frac{n}{2}+4)\}}$ is its corresponding canonical form \footnote{We use the word canonical form slightly loosely here : A top form which has (a) log singularities only on the boundary and (b) residue of this form on any face equals form on the face when thought of as a lower dimensional positive geometry. For criteria (b) to be satisfied, we need positive geometry to be such that restriction of $\widetilde{\Omega}^{Q}$ should have no exact form.}.

In Section \ref{pentinws}, we  use this intrinsic definition of world-sheet Stokes geometry for $Q=\{14,16\}$ and argue that this positive geometry is not diffeomorphic (via scattering equations) to \emph{any} of the convex realisation of the Stokes polytope in ${\cal A}_{8}$.

%


\subsection{Towards world-sheet Stokes geometries for other Quadrangulations}\label{pentinws}

In the previous section, we defined the world-sheet Stokes geometry $\widetilde{{\cal S}}^{Q}_n$ corresponding to $Q\ =\ \{14,5n,\dots,(\frac{n}{2}+1,\frac{n}{2}+4)\}$ in two equivalent ways :

\begin{enumerate}
\item  ${\cal S}_{n}^{Q}$ in ${\cal K}_{n}$ is mapped on to $\widetilde{{\cal S}}^{Q}_n$ through (inverse of) diffeomorphism generated by scattering equations. 
\item $\widetilde{{\cal S}}^{Q}_n$ is a positive geometry in the moduli space obtained by setting  $u_{mn}=0$ where $(mn)\in\ Q^{c}$. These conditions ensure that  the additional exact forms contained in $\widetilde{\Omega}^{Q}_{n}$ vanish identically. 
\end{enumerate}

For a topologically inequivalent quadrangulation such as $Q=\{14,16,\dots,(1,n-2)\}$, we did not have a criterion to identify a specific embedding of Stokes polytope inside associahedron. Thus it is not clear how to arrive at a ``canonical" definition of world-sheet Stokes geometries for such cases and consequently it is not clear how to define an integral formula similar to equation \eqref{hypecubechy}. 
 
We will now argue that point (2) can be used to define world-sheet Stokes geometries for arbitrary $Q$. We propose to define $\widetilde{{\cal S}}^{Q}_n$ for arbitrary $Q$ by demanding that the restriction of $\widetilde{\Omega}^{Q}_{n}$ on such a positive geometry has no exact form \footnote{In Section \ref{wssp14} we saw that for $Q=\{14,5n,\dots,(\frac{n}{2}+1,\frac{n}{2}+4)\}$ this restriction is in fact the canonical form on $\widetilde{{\cal S}}^{Q}_n$.}. Although we do not analyse this proposal in detail in this paper, the investigation of $n=8$ case with $Q=\{14,16\}$ in the following makes further analysis worth pursuing. 
 
Consider $\widehat{\Omega}^{\{14,16\}}_{8}$ as defined in equation \eqref{CHYform} :
\begin{equation} 
\begin{split}
     \widehat{\Omega}^{\{14,16\}}_{8} =\,& (d\log u_{14} + d\log f_{14})\wedge  (d\log u_{16} + d\log f_{16}) \\ & - (d\log u_{14} + d\log f_{14})\wedge  (d\log u_{58} + d\log f_{58})\\ & - (d\log u_{36} +d\log f_{36})\wedge  (d\log u_{16} + d\log f_{16}) \\ &+ (d\log u_{38} + d\log f_{38})\wedge  (d\log u_{58} + d\log f_{58}) \\ &+ (d\log u_{36} + d\log f_{36})\wedge  (d\log u_{38} + d\log f_{38})\,,
\end{split}{}
\end{equation}
where  the analytic functions $f_{ij}$ are given by :
\begin{equation}
\begin{split}
    f_{14} =& c_{12,4} + c_{12,5} u_{15} + c_{12,67} u_{15}u_{16} \\
    f_{38} =& c_{12,4} u_{35}u_{36} + c_{12,5} u_{36} + c_{12,67} \\
    f_{16} =& c_{12,67} + c_{3,67}u_{36} + c_{4,67} u_{36}u_{46} \\ 
    f_{58} =& c_{12,67}u_{38}u_{48} + c_{3,67}u_{48} + c_{4,67} \\ 
    f_{36} =& c_{12,4}u_{35}u_{38} + c_{12,5}u_{38} + c_{12,67}u_{16}u_{38} + c_{3,67} u_{16} + c_{4,67}u_{16}u_{46}\,,
\end{split}{}
\end{equation}
in which we introduced the notation $c_{ij,k}=c_{ik}+c_{jk}$, $c_{ij,k\ell}=c_{ik}+c_{i\ell}+c_{jk}+c_{j\ell}$ and $c_{i,jk}=c_{ij}+c_{ik}$.
As we proved in Section \ref{pfa2a}, the pushforward of $\widehat{\Omega}^{\{14,16\}}_{8}$ via scattering equations produces $m_{8}^{\{14,16\}}\ dX_{14}\ \wedge\ dX_{16}$.

We would now like to find a two-dimensional positive geometry $\widetilde{{\cal S}}^{\{14,16\}}_{8}$ in $\overline{{\cal M}}_{0,8}({\bf R})$ such that the exact form vanishes on this geometry. Without loss of generality we set $\sigma_{2}=\sigma_{1}=0$ and $\sigma_{6}=\sigma_{7}=1$, i.e. we try to locate $\widetilde{ST}^{\{14,16\}}$ in a co-dimension two boundary of the (compactified) moduli space. The two dimensional Stokes geometry can be parametrized by an equation of the form :
\begin{align}
{\cal F}(\sigma_{3},\sigma_{4},\sigma_{5})=0\,.
\end{align}
%

For generic kinematics ($\{c_{ij}\}$) the form of $f_{ij}$ is rather complicated and it is not clear to us how to analytically find such an ${\cal F}$.  However we can determine the location of the boundaries of $\widetilde{{\cal S}}^{\{14,16\}}_8$ in $\overline{{\cal M}}_{0,8}({\bf R})$ rather easily \footnote{We  note that to write an integral formula as in equation \ref{hypecubechy}, it is enough to locate the boundaries of $\widetilde{ST}^{Q}$ in $\overline{{\cal M}}_{0,8}({\bf R})$ such that $f_{ij}$ vanish on the boundary. This is because as $\widetilde{\Omega}^{Q}_{n}=\widehat{\Omega}^{Q}_{n}+\textrm{exact form}$ and hence integral of $\widetilde{\Omega}^{Q}_{n}$ on $\widetilde{ST}^{Q}$ equals the integral of $\widehat{\Omega}^{Q}_{n}$.}.

Let us first consider the $u_{16}\rightarrow 0$ facet. Residue of $\widehat{\Omega}^{\{14,16\}}_{8}$ on this facet involves  the exact form $d\ln\frac{f_{14}}{f_{36}}$.  This form vanishes for generic choice of kinematics if $f_{14}\ =\ f_{36}$. It can be easily verified using Pl\"ucker relations that this implies :
\begin{equation}\label{bnd1416ws}
\begin{array}{lll}
{\cal F}(u_{16}=0)=u_{46}\,f_{1}(u_{14},u_{15})\,,
\end{array}
\end{equation}
where $f_{1}$ is an arbitrary analytic function which remains undetermined. Similarly we can show that $\widetilde{{\cal S}}^{\{14,16\}}_8$ intersects the $u_{14}=0, u_{38}=0, u_{58} = 0$ facets of ${\cal M}_{0,8}({\bf R})$ at :
\begin{equation}\label{wsst2}
\begin{array}{lll}
{\cal F}(u_{14}=0)=u_{15}f_{2}(u_{16})\,,\\
{\cal F}(u_{38}=0)=u_{35}\,,\\
{\cal F}(u_{58}=0)=u_{48}f_{3}(u_{14})\,,
\end{array}
\end{equation}
where $f_{2}$ and $f_{3}$ remain undetermined. It can be easily seen that on the $u_{36}=0$ facet the exact form trivially vanishes.

With the boundary locations of $\widetilde{{\cal S}}^{\{14,16\}}_8$ we can see that the following holds :
\begin{equation}
\int_{\widetilde{{\cal S}}^{\{14,16\}}_{8}}\ \widetilde{\Omega}^{Q}_{8}\ \delta(X_{14}-u_{14}\ f_{14})\ \delta(X_{16}-u_{16}\ f_{16})=m_{8}^{\{14,16\}}\,.
\end{equation}
However in contrast to the world-sheet Stokes geometries for  $Q=\{14,5n,\dots,(\frac{n}{2}+1,\frac{n}{2}+4)\}$, the image of $\widetilde{{\cal S}}^{\{14,16\}}_8$ cannot be any (linear) convex realisation of $ST^{\{14,16\}}$ inside the kinematic space associahedron ${\cal A}_{8}$. We prove this as follows. If there is indeed such a convex realisation it would be determined by a linear relation :
\begin{equation}\label{lin-pent}
a_{14}X_{14}+a_{15}X_{15}+a_{16}X_{16}=C\,,
\end{equation}
where $a_{ij}$ and $C$ are positive constants. However from equations \eqref{bnd1416ws} and \eqref{wsst2} these constants are constrained.

At $u_{14}=0$ boundary, $X_{14}\ \rightarrow\ 0$ and from equation \eqref{wsst2} we have that $X_{15}\ \rightarrow\ 0$. Thus equation \eqref{lin-pent} becomes :
\begin{equation}
X_{16}=\textrm{const}\,,
\end{equation}
which is meaningless, as $X_{14}\ \rightarrow\ 0$ is a one dimensional facet of the  pentagon ${\cal S}^{\{14,16\}}_{8}$. 

As we saw in Section \ref{proof-convergentdissection}, there is no canonical embedding of a Stokes polytope with $Q=\{14,16\}$ in ${\cal A}_{8}$. We now see that from the perspective of world-sheet and scattering equations we do not get any linear realisation of Stokes polytope in kinematic space, i.e. if we seek a world-sheet Stokes geometry on which the exact form vanishes then even though there does exist such a positive geometry in $\overline{{\cal M}}_{0,8}({\bf R})$, it does not get mapped to \emph{any} convex realisation of Stokes polytope in kinematic space ! Although our data point is merely one example, we believe the lesson drawn is more general and will be valid for arbitrary $n$ and any quadrangulation other than $Q=\{14,5n,\dots,(\frac{n}{2}+1,\frac{n}{2}+4)\}$. Whether there exists diffeomorphisms which map $\widetilde{{\cal S}}^{Q}_n$ to \emph{a} convex realisation of Stokes polytope in kinematic space is a question we leave for future investigation.

\section{Outlook}

The CHY formalism  and the  Amplituhedron program are two of the most outstanding recent developments in our conceptual understanding of S Matrix theory. For bi-adjoint scalar $\phi^{3}$ amplitudes, these paradigms are two avatars of the same underlying object, namely, associahedron and its corresponding canonical form. However for generic non-supersymmetric QFTs, the relationship between CHY formula and positive geometries program is still in its nascent stages if we move beyond cubic vertices. In this paper we tried to analyse just such a relationship for planar quartic scalar interactions. 

A key result that was obtained thanks to recent developments in the theory of cluster algebras, quivers and accordiohedra is that $\phi^{4}$ amplitudes can be obtained via lower projective forms on kinematic space associahedra. The subset of the entire class of ABHY associahedra defined in \cite{HughThomas} were used to deduce this fact. From this perspective, only the combinatorial data of a Stokes polytope (defined via a reference quadrangulation $Q$) was required to define a projective lower form on the associahedron. Although our analysis was for quartic interactions, we believe it can be readily generalised to $\phi^{p}$ where $p>4$ interactions with projective forms of appropriate ranks and their restriction to ABHY associahedra generating the corresponding amplitudes. One may wonder if we can classify all such scalar theories in terms of projective forms of varying ranks on ABHY associahedra. We leave such speculation for future investigation.

The above mentioned result was then transcribed to the world-sheet where there are indeed lower forms on the CHY moduli space whose pushforward via scattering equations produced $m^{Q}_{n}$. We also showed that these lower forms on kinematic space associahedra are canonical forms on Stokes polytopes which admit a realisation in the interior or on the boundary of associahedra. 

Our analysis is rather preliminary. However it brings to the fore lower ranked $d\ln$ forms (and corresponding cohomology groups) on the CHY moduli space which generate scattering amplitudes of (planar) scalar interactions. We conclude with certain open questions which may be worth investigating. 

Focussing on projective lower forms on ABHY associahedra offers a new possibility to revisit the issue of weights $\alpha_{Q}$ which appeared in equation \eqref{weights1}. Computation of $\alpha_{Q}$ (which only depends on $Q$ up to cyclic permutation) was performed in \cite{Banerjee:2018tun},\cite{Raman:2019utu} by demanding that when summed over all $Q$'s, each channel contributes to $m_{n}$ with a unit residue. This criterion relied on input from outside the domain of polytopes and canonical forms to compute the amplitude and that was philosophically a step back from the tenets of the Amplituhedron program (where the geometry of the polytopes is enough to compute amplitude of the theory which turned out to be unitary and local). The problems originated from the fact that when viewed as canonical forms on Stokes polytopes, each $\Omega^{Q}_{n}$ has unit normalisation. But when viewed as lower forms on associahedra, the normalisation of these forms is not determined apriori and opens up a possibility to define it via some new criterion which may depend on $Q$ and ${\cal A}_{n}^{T}$.  

As we argued in Section \ref{pentinws}, for $Q\ \neq\ \{14,5n,\dots,(\frac{n}{2}+1,\frac{n}{2}+4)\}$ although there may exist world-sheet Stokes geometries in ${\cal M}_{0,n}({\bf R})$ on which CHY-inspired integral formulae could be written down to compute $m_{n}^{Q}$, these positive geometries are not diffeomorphic to \emph{any} linear realisation of Stokes polytopes inside the kinematic space. A detailed analysis of such positive geometries and investigation of diffeomorphisms which map them to kinematic space Stokes polytope is required to write CHY type integral formulae for arbitrary quandrangulations and perhaps use to define a moduli space whose polytopal realisations are Stokes polytopes.
 
It is also important to investigate the relationship of the $\frac{n-4}{2}$ forms $\sum_{Q}\ {\alpha}_{Q}\ \widehat{\Omega}^{Q}_{n}$ with the $n-3$ forms for $\phi^{4}$ theory defined in \cite{Baadsgaard:2015ifa}. As these forms have singularities on the same boundaries of ${\cal M}_{0,n}({\bf R})$, we believe that a precise relationship between them must exist but may require further conceptual inputs than the ones given in this paper.

Finally, in \cite{Kalyanapuram:2019nnf}, a different outlook on the world-sheet perspective was provided by computing intersection numbers for Stokes polytopes in the moduli space. We believe that the pushforward maps derived in this paper are rather closely tied to the ideas advocated in \cite{Kalyanapuram:2019nnf}, and a more precise analysis of this relationship is worth pursuing.

\vskip 1.5cm
\noindent {\large {\bf Acknowledgments}}
\vskip 0.2cm
\noindent We are grateful to Sujay K. Ashok, Anirban Basu, Miguel Campiglia, Subhroneel Chakrabarti, Poul Damgaard, Abhijit Gadde, Nikhil Kalyanapuram, Akavoor Manu,  Prashanth Raman and Arnab Priya Saha for discussions and critical inputs. We are extremely grateful to Nikhil Kalyanapuram for making his manuscript available to us prior to publication.  We are indebted to Alfredo Guevara, Song He, Dileep P. Jatkar, Madhusudhan Raman and Ashoke Sen for a number of insightful discussions, and Song He for constant encouragement and his generosity with regards to time.  We would also like to thank Frederic Chapoton, Satyan Devadoss and Sushmita Venugopalan for a number of pertinent clarifications and to Frederic Chapoton for pointing out the key reference \cite{1906ppp}. AL is grateful to Sarjick Bakshi, KV Subrahmanyam and Sukendu Merhotra for conducting the Cluster Algebra seminar at CMI and to Sarjick Bakshi and Suratno Basu for being patient with his numerous doubts. AL is thankful to ICTS Bangalore, HRI Allahabad, Munich Institute for Astro- and Particle Physics (MIAPP), Ashoka University, BITS-Pilani (Goa campus) and especially IIT Gandhinagar for their kind hospitality during various stages of this project. The work of AL was partially supported by Scholars in Residence program at IIT Gandhinagar. The work of R.R.J is supported by the MIUR PRIN Contract 2015 MP2CX4 ``Non-perturbative Aspects Of Gauge Theories And Strings". The work of R.R.J is also partially supported by ``Fondi Ricerca Locale dell'Universit\`a del Piemonte Orientale". PB would like to acknowledge the support provided by the Max Planck Partner Group grant MAXPLA/PHY/2018577.

\noindent
%
\vskip 1cm

\appendix

\section{Notation}
We collect the relevant notation that goes into the paper below :
\begin{itemize}
\item $\mathcal{K}_{n}$ :  Kinematic space for $n$ massless momenta
\item $T$ : Triangulation of an $n$-gon
\item $Q$ :  Quadrangulation of an $n$-gon 
\item $\mathcal{A}_n$ :  ABHY associahedron in kinematic space associated to $T=\{(13),(14),\ldots\,(1,n-1)\}$  
\item $\mathcal{A}_n^T$ :  ABHY associahedron in kinematic space associated to any other triangulation 
\item $\Omega^Q_n$ :  Projective planar scattering form associated to quadrangulation $Q$ of an $n$-gon
\item $\mathcal{S}^Q_n$ :  Stokes polytope in kinematic space associated to quadrangulation $Q$ of an $n$-gon  
\item $m_n^Q$ :  Canonical rational function/ partial amplitude associated to $\mathcal{S}^Q_n$
\item $m_n$ :  Tree level planar $n$-point scattering amplitude
\item $\mathcal{M}_{0,n}({\textbf R})$ : Real section of the moduli space of  genus 0 with $n$ punctures
\item $\overline{\mathcal{M}_{0,n}}({\textbf R})$ : Compactification of $\mathcal{M}_{0,n}({\textbf R})$
\item $\widetilde\Omega^Q_n$ : Projective form on the moduli space corresponding to $\Omega^Q_n$
\item $\widetilde{{\cal S}}^{Q}_n$ :  world-sheet Stokes geometry associated to quadrangulation $Q$
\item $\widetilde\omega^Q_n$ : Restriction of $\widetilde\Omega^Q_n$ on to $\widetilde{{\cal S}}^{Q}_n$
\end{itemize}


\section{Convex realisation of Stokes polytopes : An example}\label{appenalgo}
Algorithm given in Section \ref{sec3} gives us a geometric realization of Stokes polytope inside the kinematic space. Here we will look at an example of geometric realization of Stokes polytope with $n=8$ and reference quadrangulation $(1,4),(1,6)$.

\begin{itemize}
    \item We first draw hollow vertices $\widetilde{v}_1,\widetilde{v}_3,\widetilde{v}_4,\widetilde{v}_5,\widetilde{v}_6$ and $\widetilde{v}_8$ on the edges $(1,2),(3,4),(4,5),(5,6),(6,7)$ and $(8,1)$ respectively. We place a hollow vertex $v_{1,4}$ on the diagonal $D_{1,4}$ and a hollow vertex $v_{1,6}$ on the diagonal $D_{1,6}$.
    \begin{figure}[H]
    \centering
    \includegraphics[scale=0.3]{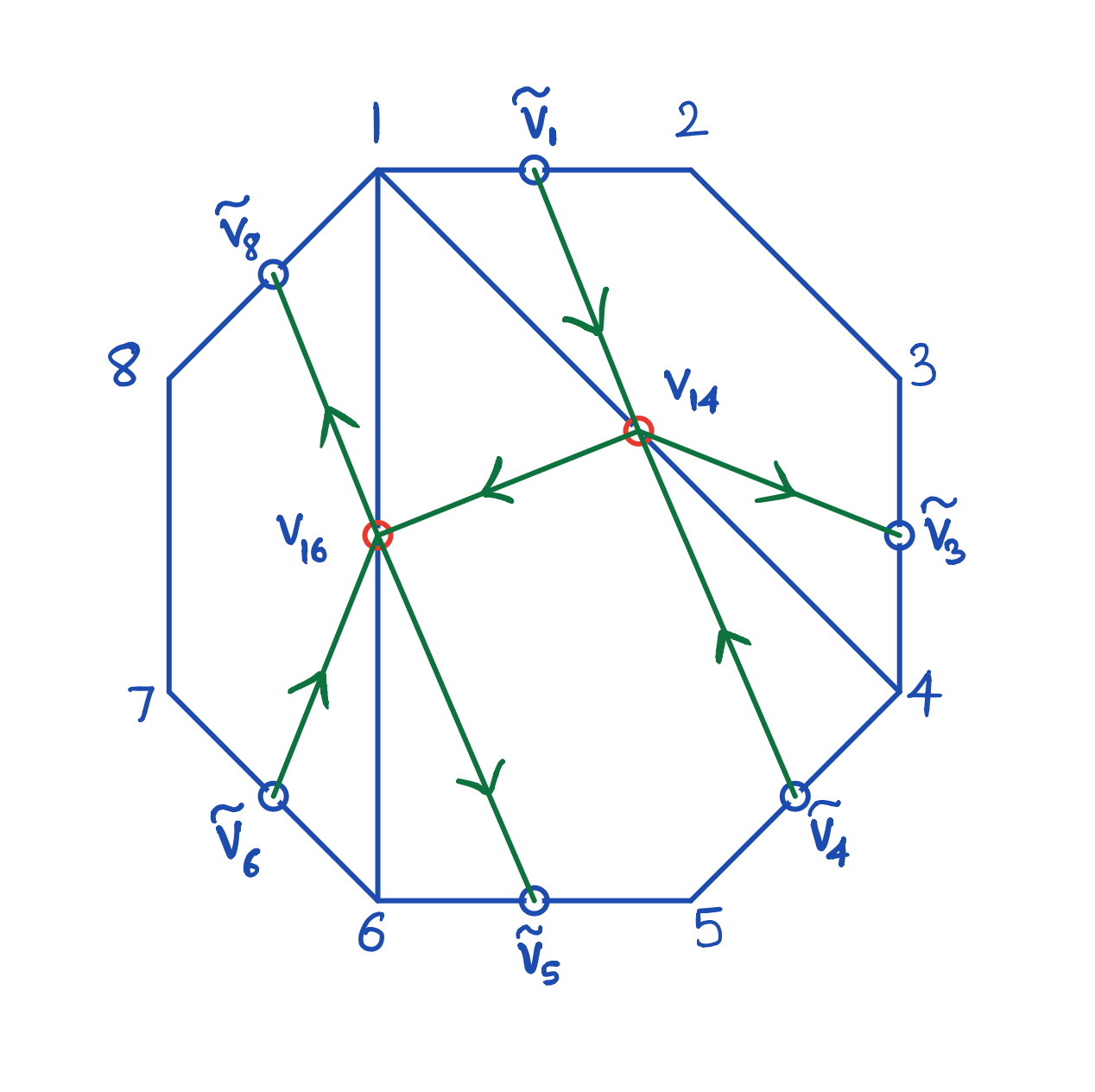}
    \end{figure}
    \item  Now we draw the following arrows, $(\widetilde{v}_1,v_{1,4}),(v_{1,4},\widetilde{v}_{3}),(\widetilde{v}_4,v_{1,4}),(v_{1,6}\widetilde{v}_{5}),(\widetilde{v}_{6},v_{1,6}),(v_{1,6},\widetilde{v}_{8})$ and $(v_{1,4},v_{1,6})$. Here the direction of arrow $(a,b)$ is from $a$ to $b$.
    \item For this quadrangulation there are three internal paths, two paths of length zero and a path of length one. Let $p_1 =v_{1,4} $, $p_2= v_{1,6}$ and $p_3 = v_{1,4}v_{1,6} $.
    \item The subset of proper walks that appear in the algorithm given in Section \ref{sec3} are :
    \begin{align*}
        \hook p_1 \unhook & = \widetilde{v}_3 v_{1,4} v_{1,6} \widetilde{v}_{6} &  \cohook p_1 \counhook &= \widetilde{v}_1 v_{1,4} \widetilde{v}_{4} &   
        \hook p_1 \counhook & = \widetilde{v}_3 v_{1,4}\widetilde{v}_{4}  & 
        \cohook p_1 \unhook & = \widetilde{v}_1 v_{1,4} v_{1,6} \widetilde{v}_{6}.\\
         \hook p_2 \unhook & = \widetilde{v}_5 v_{1,6} \widetilde{v}_{8} &  \cohook p_2 \counhook &= \widetilde{v}_3 v_{1,4} v_{1,6} \widetilde{v}_{6} &  
         \hook p_2 \counhook & = \widetilde{v}_5 v_{1,6}\widetilde{v}_{6}  &  
         \cohook p_2 \unhook & = \widetilde{v}_3 v_{1,4} v_{1,6} \widetilde{v}_{8}.\\
           \hook p_3 \unhook & = \widetilde{v}_{3}v_{1,4} v_{1,6} \widetilde{v}_{8} &  \cohook p_3 \counhook &= \widetilde{v}_{1}v_{1,4} v_{1,6} \widetilde{v}_{6} &  \hook p_3 \counhook & = \widetilde{v}_{3}v_{1,4} v_{1,6} \widetilde{v}_{6} &  \cohook p_3 \unhook & = \widetilde{v}_{1}v_{1,4} v_{1,6} \widetilde{v}_{8} &.   
    \end{align*}
    \end{itemize}
Therefore the constraints for Stokes polytope with reference quadrangulation $(1,4),(1,6)$ are 
\begin{align}\label{1416hch}
    X_{36} + X_{14} - X_{16} &= d_{p_1}\,,\cr
    X_{58} + X_{36} - X_{38} &= d_{p_2}\,,\cr
    X_{38} + X_{16} - X_{36} &= d_{p_3}\,.
\end{align}
This algorithm captures the notions of $Q$-compatibility in Stokes polytope and compatibility in a general accordiohedron. A diagonal $D_{i,j}$ is compatible with given reference dissection if and only if there exists a proper walk between hollow vertices $\widetilde{v}_i$ and $\widetilde{v}_j$. To see this we first notice that the proper walk between hollow vertices $\widetilde{v}_i$ and $\widetilde{v}_j$, with $j \neq i \pm 1 $ can be continuously deformed to get the hollow diagonal $D^{\circ}_{i,j}$. Therefore, the vertices of the proper walk are precisely the vertices on the diagonals which intersect the hollow diagonal $D^{\circ}_{i,j}$. Since the arrows are drawn only between the diagonals and edges which intersect, the set of edges $[i,i+1],[j,i+1]$ and diagonals which intersect the hollow diagonal $D^{\circ}_{i,j}$ is connected. On the other hand if the set of edges $[i,i+1],[j,i+1]$ and diagonals which intersect the hollow diagonal $D^{\circ}_{i,j}$ is connected it is easy to see that we can construct a proper walk between the hollow vertices $\widetilde{v}_i$ and $\widetilde{v}_j$.

\section{Convex realisation of Stokes Polytopes via an equivalent set of Constraints}\label{newappfor321}

In this appendix a second set of constraints which locate Stokes polytope in kinematic space \footnote{These set of constraints are derived in theorem {\bf 2.33} in \cite{1906ppp}. The complete derivation of these constraints involves describing Stokes Polytopes via ${\bf g}$-vectors.}.

Let $\mathcal{P}_{k\ell}$ be the set of all hollow vertices (paths of length zero) at which the proper walk between vertices $\widetilde{v}_{k}$ and $\widetilde{v}_{l}$ peaks, and let $\mathcal{V}_{k\ell}$ be the set of all hollow vertices at which the proper walk between the two vertices deeps. It can be shown that the constraints given in equation (\ref{stokc}) are equivalent to the following :
\begin{align}\label{gvect}
    X_{ij} = \sum_{(pq) \in \mathcal{V}_{ij} } X_{pq} - \sum_{(pq) \in \mathcal{P}_{ij}} X_{pq} + \epsilon_{ij}\,,
\end{align}
for all $X_{ij}$ compatible with the reference quadrangulation.
Here $\epsilon_{ij}$ are linear combinations of the $d_{p_{k}}$ in equation (\ref{stokc}).

These constraints appear to be very different then the ones defined in equation \eqref{stokc}. While the previous set of constraints were $\vert K\vert$ in number which (if the reference quadrangulation $Q$ does not contain any parallel diagonals) equals the number of $Q$-compatible planar variables minus the dimension of the Stokes polytope $\frac{n-4}{2}$, this set of constraints have cardinality equal to the number of $Q$-compatible $X_{ij}$. However if $(ij)\ \in\ Q$ then these constraints are trivially satisfied. 
 
 The advantage of writing the constraints this way is that it gives all the compatible $X_{ij}$ in terms of diagonals of the reference quadrangulation.

We derive these set of constraints for $n=8$ case and $Q\ =\ \{14,16\}$ to show that they are indeed equivalent to those given in equation \eqref{1416hch}. The proper walk between the vertices $\widetilde{v}_{3}$ and $\widetilde{v}_{6}$ is $\widetilde{v}_{3}, v_{1,4}, v_{1,6}, \widetilde{v}_{6}$. It peaks at hollow vertex $v_{1,4}$ and deeps at hollow vertex $ v_{1,6}$. Therefore the sets $\mathcal{P}_{3,6}$ and $\mathcal{V}_{3,6}$ are given by $\mathcal{P}_{3,6} =  \{v_{1,4} \} $ and $\mathcal{V}_{3,6} = \{ v_{1,6} \}$. Similarly, we can find the sets $\mathcal{P}_{k\ell}$ and $\mathcal{V}_{k\ell}$ for other proper walks. 
\begin{align}
\mathcal{P}_{1,4} & = \emptyset &  \mathcal{V}_{1,4} & = \{ v_{1,4} \}   \\
\mathcal{P}_{3,8} & = \{ v_{1,4} \} &  \mathcal{V}_{3,8} & =  \emptyset \\
\mathcal{P}_{1,6} & = \emptyset &  \mathcal{V}_{1,6} & = \{ v_{1,6} \}  \\
\mathcal{P}_{5,8} & =  \{ v_{1,6} \} &  \mathcal{V}_{5,8} & = \emptyset  
\end{align}
Therefore we get the following constraints from equation (\ref{gvect}),
\begin{align}\label{1416gvect}
X_{14} &= X_{14} &
 X_{16} &= X_{16} \\
 X_{38} &= -X_{14} + \epsilon_{38} &
 X_{58} &= -X_{16} + \epsilon_{58} \\
 X_{36} &= X_{16} - X_{14} + \epsilon_{36} & & 
\end{align}
Now it is easy to see that with $\epsilon_{38} = d_{p_1} + d_{p_3} $, $\epsilon_{58} = d_{p_2} + d_{p_3} $ and $\epsilon_{36} =  d_{p_1}$ the above equation \eqref{1416gvect} are equivalent to the equation \eqref{1416hch}. 

We use the above constraints to restrict the planar scattering form onto Stokes polytope. In Appendix \ref{spoofQ} we show that this restriction is proportional to the partial quartic $n$-point amplitude $m_{n}^{Q}$.
%
%
%

\section{Analysing geometric constraints on $X_{ij}\ (ij)\in Q$}\label{appb}
In this appendix we will show that the Stokes polytope constraints in equation (\ref{stokc}) are linear combinations of associahedron constraints in equation (\ref{sijcijT}). We will consider the set of constraints given in equation (\ref{gvect}). These equations provide a convex realisation of accordiohedra, in particular they provide a convex realisation of Stokes polytope and associahedron. We will show that the set of constraints given in equation (\ref{gvect}) for Stokes polytope is a subset of set of constraints given in equation (\ref{gvect}) for associahedron. Since, the constraints in equation (\ref{gvect}) for Stokes polytope are equivalent to constraints in equation (\ref{stokc}) and the constraints in equation (\ref{gvect}) for associahedron are equivalent to constraints in equation (\ref{sijcijT}), we would have proved our claim.

For a given $X_{ij}$ compatible with the reference dissection $D$ we denote by $\mathcal{P}_{D,i,j}$  the set of all hollow vertices at which the proper walk between the vertices $\widetilde{v}_{i}$ and $\widetilde{v}_{j}$ peaks and similarly we denote by $\mathcal{V}_{D,i,j}$ the set of all hollow vertices at which the proper walk between the vertices $\widetilde{v}_{i}$ and $\widetilde{v}_{j}$ peaks. 

\begin{claim}\label{subpeak}
If $D_{1}$ and $D_{2}$ are two dissections of an $n$-gon such that $D_{1} \subset D_{2} $, then $\mathcal{P}_{D_{1},i,j} = \mathcal{P}_{D_{2},i,j} $ and $\mathcal{V}_{D_{1},i,j} = \mathcal{V}_{D_{2},i,j} $ for all $(i,j)$ compatible with reference dissection $D_{1}$.
\end{claim}{}
\begin{proof}
We will show that for an arbitrary dissection $D$ of an $n$-gon if we remove a diagonal $(k,l) \in D $ from $D$ then $\mathcal{P}_{D\setminus (k,l),i,j} = \mathcal{P}_{D,i,j} $ and $\mathcal{V}_{D\setminus (k,l),i,j} = \mathcal{V}_{D,i,j} $ for all $(i,j)$ compatible with reference dissection $D\setminus (k,l)$. 

Given a diagonal $(k,l)\in D$, the arrows and vertices near the diagonal $(k,l)$ look like figure \ref{withkl}, where $(k,p),(l,q),(l,r)$ and $(k,s)$ are sides or diagonals in $D$. Any walk containing $v_{k,l}$ contains one of the following paths; $v_{s,k}v_{k,l}v_{k,p}$,$v_{l,r}v_{k,l}v_{l,q}$,$v_{s,k}v_{k,l}v_{l,q}$ or $v_{k,p}v_{k,l}v_{l,r}$. The walks containing $v_{s,k}v_{k,l}v_{l,q}$ have a peak at $(k,l)$ and walks containing $v_{k,p}v_{k,l}v_{l,r}$ have a deep at $(k,l)$.

\begin{figure}[H]
    \centering
    \includegraphics[scale=0.40]{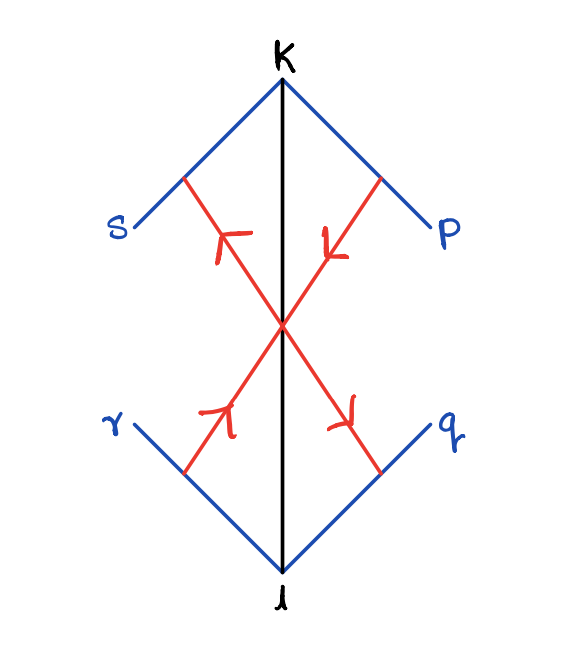}
    \caption{Arrows and vertices near the diagonal $(k,l)$}
    \label{withkl}
\end{figure}

Now if we remove the diagonal $(k,l)$ from the reference dissection, that is if consider the reference dissection $D\setminus (k,l)$, the arrows and diagonals look like figure \ref{withoutkl}. The walks in reference $D$ containing $v_{s,k}v_{k,l}v_{k,p}$ or $v_{l,r}v_{k,l}v_{l,q}$ reduce to walks containing $v_{s,k}v_{k,p}$ or $v_{l,r}v_{l,q}$ respectively. While the walks containing $v_{s,k}v_{k,l}v_{l,q}$ or $v_{k,p}v_{k,l}v_{l,r}$ can not be reduced to walks of $D\setminus (k,l)$. In other words if $(i,j)$ is compatible with dissection $D$ then $(i,j)$ is compatible with the dissection $D\setminus (k,l)$ if and only if the walk in $D$ between vertices $\widetilde{v}_i$ and $\widetilde{v}_j$ do not peak or deep at $(k,l)$.

\begin{figure}[H]
    \centering
    \includegraphics[scale=0.40]{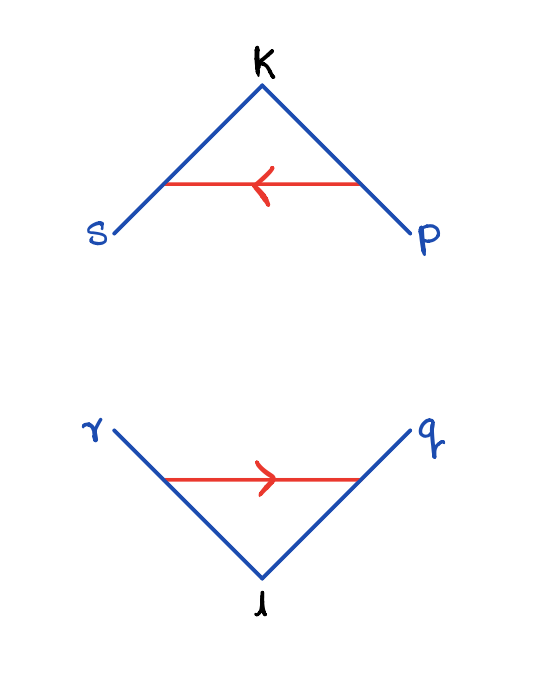}
    \caption{Arrows and vertices without diagonal $(k,l)$ }
    \label{withoutkl}
\end{figure}

Since the directions of arrows remain same in the reduced walks, other peaks and deeps of walks in $D$ are same as peaks and deeps of walks in $D\setminus(k,l)$. Therefore 
\begin{align}
    \mathcal{P}_{D,i,j} &=  \mathcal{P}_{(D\setminus(k,l)),i,j} &  \mathcal{V}_{D,i,j} &=  \mathcal{V}_{(D\setminus(k,l)),i,j},
\end{align}
for all $(i,j)$ compatible with the dissection $D\setminus(k,l)$.

Now, removing diagonals one by one it is easy to see that, if $D_{1}$ and $D_{2}$ are two dissections of an $n$-gon such that $D_{1} \subset D_{2} $, then $\mathcal{P}_{D_{1},i,j} = \mathcal{P}_{D_{2},i,j} $ and $\mathcal{V}_{D_{1},i,j} = \mathcal{V}_{D_{2},i,j} $ for all $(i,j)$ compatible with reference dissection $D_{1}$.
\end{proof}{} 

The convex realisations of accordiohedra are given by 
\begin{equation}
    X_{ij} = \sum_{(p,q) \in \mathcal{V}_{i,j} } X_{p,q} - \sum_{(p,q) \in \mathcal{P}_{i,j}} X_{p,q} + \epsilon_{i,j},
\end{equation}{}
for all $X_{ij}$ compatible with given reference dissection. Using claim \ref{subpeak} it is easy to see that the set of constraints for Stokes polytope is a subset of set of constraints for associahedron.

\section{Canonical form on kinematic space and scattering amplitude : Proof }\label{spoofQ}

In this section we review and generalise the results of \cite{Banerjee:2018tun} and show that given a reference quadrangulation $Q$ and the corresponding canonical form $\Omega^{Q}$ in kinematic space, the  induced form on ${\cal S}^{Q}_{n}$ produces a partial contribution to $\phi^{4}$ scattering amplitude.

\begin{claim}
Given a reference dissection $D$, suppose $X_{ij}$ and $X_{kl}$  are mutations of each other, i.e. $X_{ij}$ and $X_{kl}$  are  compatible intersecting diagonals and there exists a set of compatible diagonals $S$ such that $\{X_{ij}\} \cup S $ and $\{X_{k\ell}\} \cup S $ are maximal compatible dissections.  Let $\Omega_{S} =  \bigwedge_{X_{a,b} \in S} \mathrm{d}X_{a,b} $. Then 
\begin{equation}
    \mathrm{d} X_{ij} \wedge \Omega_{S} = -  \mathrm{d} X_{k\ell} \wedge \Omega_{S}.
\end{equation}{}
\end{claim}{}

\begin{proof}

Given any intersecting diagonals $X_{ij}$ and  $X_{k\ell}$ with $i<k<j<l$, they divide the vertices of the polygon in four sets, viz. $V_{i,k}=\{ i,i+1,\ldots, k \}$, $V_{k,j}=\{ k,k+1,\ldots, j \}$, $V_{j,l}=\{ j,j+1,\ldots, l \}$ and $V_{l,i}=\{ l,l+1,\ldots, i \}$. Any diagonal which does not intersect $X_{ij}$ and $X_{k\ell}$ has both the end points in one of the above four sets. That is, if $X_{p,q}$ does intersect $X_{ij}$ and $X_{k\ell}$ then both $p$, $q$ belong to one of $V_{i,k}$,$V_{k,j}$,$V_{j,l}$ or $V_{l,i}$.

Notice $X_{ik},X_{kj},X_{j\ell}$ and $X_{l,i}$ do not intersect with any diagonal whose both the end points are in one of $V_{i,k}$,$V_{k,j}$,$V_{j,l}$ or $V_{l,i}$, other than themselves. Therefore, if $X_{p,q} \in \{ X_{ik},X_{kj},X{j,l}, X_{l,i}  \} $  is a compatible diagonal then $ X_{p,q} \in S $.

$X_{ij}$ and $X_{k\ell}$ intersect, therefore the proper walk between $\widetilde{v}_{i}$ and $\widetilde{v}_{j}$, $W_{i,j}$ and the proper walk between $\widetilde{v}_{k}$ and $\widetilde{v}_{l}$, $W_{k,l}$ have a path $\rho$ in common such that one of $W_{i,j}$ and $W_{k,l}$ peaks at $\rho$ and the other one deeps at $\rho$. Without loss of generality let's assume $W_{k,l}$ peaks at $\rho$ and $W_{i,j}$ deeps at $\rho$.

\textbf{Case I :} 

Suppose the length of $\rho$ is greater than zero. That is, $\rho$ is not just a vertex. Suppose $\rho$ starts at vertex $v_{a,b}$ and ends at vertex $v_{c,d}$ , $\rho = v_{a,b} \widetilde{\rho} v_{c,d} $. Let $\alpha_{i},\alpha_{j},\alpha_{k},\alpha_{l}$ be such that
\begin{align}
    W_{i,j} &= \alpha_{i}  v_{a,b} \widetilde{\rho} v_{c,d} \alpha_{j} &    W_{k,l} &= \alpha_{k}  v_{a,b} \widetilde{\rho} v_{c,d} \alpha_{l}.
\end{align}
Consider walks $W_{i,l} = \alpha_{i}  v_{a,b} \widetilde{\rho} v_{c,d} \alpha_{l} $ and   $W_{k,j} = \alpha_{k}  v_{a,b} \widetilde{\rho} v_{c,d} \alpha_{j} .$ $W_{i,l}$ is a proper walk if and only if it peaks or deeps at some vertex in $W_{i,l}$, similarly $W_{k,j}$ is a proper walk if and only if it peaks or deeps at some vertex in $W_{k,j}$. 

Now, we note the following obersvations,
\begin{itemize}
     \item A walk in $ \{ W_{i,j}, W_{k,l},W_{i,l},W_{k,j}  \} $ peaks/deeps at a vertex $v$ in $\widetilde{\rho}$ if and only if all other walks in $ \{ W_{i,j}, W_{k,l},W_{i,l},W_{k,j}  \} $ peak/deep at $v$.
    \item $W_{i,j}$ peaks/deeps at a vertex $v \in \alpha_{i} v_{a,b} $ if and only if $W_{i,l}$ peaks/deeps at v.
    \item $W_{i,j}$ peaks/deeps at a vertex $v \in v_{c,d}\alpha_{j} $ if and only if $W_{k,j}$ peaks/deeps at v.
    \item $W_{k,l}$ peaks/deeps at a vertex $v \in \alpha_{k} v_{a,b} $ if and only if $W_{k,j}$ peaks/deeps at v.
    \item $W_{k,l}$ peaks/deeps at a vertex  $v \in v_{c,d}\alpha_{l} $ if and only if $W_{i,l}$ peaks/deeps at v.   
\end{itemize}

Therefore using the constraints given in equation (\ref{gvect}) it is easy to see that 
\begin{equation}
    X_{ij} +  X_{k\ell} - X_{i,l} - X_{kj} =   \epsilon_{i,j} +    \epsilon_{k,l} -   \epsilon_{i,l} -   \epsilon_{k,j}. 
\end{equation}{}

In case $X_{i,l}$ or $X_{kj}$ are not compatible diagonals that is,  $W_{i,l}$ or $W_{k,j}$ are not proper walks, we can take $X_{i,l} = \epsilon_{i,l}=0 $ or $X_{kj} = \epsilon_{k,j}=0 $ and the above equation will hold.

\textbf{Case II:}

Suppose $\rho$ is just a vertex. Let $\alpha_{i},\alpha_{j},\alpha_{k},\alpha_{l}$ be such that
\begin{align}
    W_{i,j} &= \alpha_{i} \rho \alpha_{j} &    W_{k,l} &= \alpha_{k}  \rho \alpha_{l}.
\end{align}
Consider walks $W_{i,l} = \alpha_{i} \rho \alpha_{l} $ and   $W_{k,j} = \alpha_{k} \rho \alpha_{j} .$ $W_{i,l}$ and $W_{k,j}$ are proper walks if and only if they peak or deep at some vertex. 

Now, we note the following observations, 
\begin{itemize}
    \item $W_{i,j}$ peaks/deeps at a vertex $v \in \alpha_{i} $ if and only if $W_{i,l}$ peaks/deeps at v.
    \item $W_{i,j}$ peaks/deeps at a vertex $v \in \alpha_{j} $ if and only if $W_{k,j}$ peaks/deeps at v.
    \item $W_{k,l}$ peaks/deeps at a vertex $v \in \alpha_{k} $ if and only if $W_{k,j}$ peaks/deeps at v.
    \item $W_{k,l}$ peaks/deeps at a vertex  $v \in \alpha_{l} $ if and only if $W_{i,l}$ peaks/deeps at v.   
\end{itemize}

$W_{i,l}$ and $W_{k,j}$ do not peak or deep at $\rho$ but $W_{i,j}$ deeps at $\rho$ and $W_{k,l}$ peaks at $\rho$ therefore just as in case I we have 
\begin{equation}
    X_{ij} +  X_{k\ell} - X_{i,l} - X_{kj} =   \epsilon_{i,j} +    \epsilon_{k,l} -   \epsilon_{i,l} -   \epsilon_{k,j}. 
\end{equation}{}

In case $X_{i,l}$ or $X_{kj}$ are not compatible diagonals that is,  $W_{i,l}$ or $W_{k,j}$ are not proper walks, we can take $X_{i,l} = \epsilon_{i,l}=0 $ or $X_{kj} = \epsilon_{k,j}=0 $ and the above equation will hold.

Therefore, $\mathrm{d}(X_{i,j} + X_{k,l} ) = \mathrm{d} X_{i,l} +  \mathrm{d} X_{k,j} $. Whenever $X_{i,l}$ and $X_{k,j}$ are compatible they belong to $S$ and when $X_{i,l}$ is not compatible $\mathrm{d}X_{i,l}=0$, and when $X_{k,j}$ is not compatible $\mathrm{d} X_{k,j}  = 0$. Hence,  
\begin{equation}\label{cfproof}
    \mathrm{d}( X_{i,j}+ X_{k,l}) \wedge \Omega_{S} = 0.
\end{equation}{}
\end{proof}{}

The planar scattering form is given by 
\begin{equation}
\Omega_{n}^{D} = \sum_{D'} \text{sign}(D') \left( \prod_{(i,j)\in D'} \frac{1}{X_{i,j}} \right) \bigwedge_{(i,j)\in D'} \mathrm{d}X_{i,j},
\end{equation}
where the sum is over dissection compatible with the reference dissection $D$.
Now, using equation (\ref{cfproof}) it is easy to see that the pull back of the planar scattering form on the accordiohedron is given by 
\begin{equation}
\omega^{D} = \sum_{D'} \left( \prod_{(i,j)\in D'} \frac{1}{X_{i,j}} \right) \bigwedge_{(i,j)\in D} \mathrm{d}X_{i,j}
\end{equation}
Thus, the induced form on the accordiohedron produces a partial contribution to the scattering amplitude. In particular, the induced form on ${\cal S}_{n}^{Q}$ produces a partial contribution to the $\phi^4$ scattering amplitude.  
\section{Map from world-sheet associahedron to ABHY associahedron}\label{messscatt}
In this appendix, we will propose a map from the world-sheet associahedron to ABHY associahedron with arbitrary reference triangulation. In \cite{Arkani-Hamed:2017mur} Arkani Hamed et al. gave a map from the world-sheet associahedron to ABHY associahedron with reference $T=\{(13),(14),\ldots,(1,n-1)\}$ using the scattering equations, but there is nothing special about this reference triangulation and we expect that the scattering equations will give a map from world-sheet associahedron to any ABHY associahedron.The most rigourous way to arrive at a diffeomorphism between $\overline{{\cal M}}_{0,n}({\bf R})$ and an ABHY associahedron is to generalise the derivation in \cite{Arkani-Hamed:2017mur} suitably. We believe this is possible but instead of attempting such a rigourous proof and propose a map (based on certain arguments given below) which we believe is the map induced by scattering equations.  

Given a triangulation $T$, using equation \eqref{gvect} we can express any planar variable $X_{mn}$ in terms of $\{X_{k\ell} \vert (k,\ell) \in T\}$ and $\bigcup_{(k,\ell) \in T} \{ c_{p,q} \vert \text{the diagonal} (p,q) \text{ intersects the diagonal } (k-1,\ell-1) \} $.  Therefore, once we have a map from the world-sheet associahedron to $X_{ij}$ for $(i,j) \in T$ we can get the map from world-sheet associahedron to ABHY Associahderon. The map to $X_{ij}$ for $(i,j) \in T$ should be of the form 
\begin{equation}
    X_{ij} = \sum_{(k,\ell) \in \mathcal{C}_{ij} }  c_{k\ell} f_{i,j,k,\ell}(\sigma).  
\end{equation}{}
Here we expect the set $\mathcal{C}_{ij}$ to depend only on $(i,j)$ and not the reference triangulation. Hence, the natural choice for $\mathcal{C}_{ij}$ is $\mathcal{C}_{ij} =  \{ c_{p,q} \vert \text{the diagonal } (p,q) \text{ intersects the diagonal} (i-1,j-1)\} $. We expect that the function 
$f_{i,j,k,\ell}(\sigma)$ is a $\mathbf{SL}(2,\mathbb{C})$ invariant function of $\sigma$s which depends only on $(i,j,k,\ell)$. A natural choice of such a function is the cross ratio $f_{i,j,k,\ell}(\sigma) = \frac{\sigma_{i-1,\ell }\sigma_{k,j-1}}{\sigma_{k,\ell} \sigma_{i-1,j-1}}$. 
Now we note the following identity which relates Pl\"ucker coordinates to the $\sigma$ coordinates,
\begin{equation}
    \prod_{\substack{a\leq p \leq b\\ c\leq q \leq d}} u_{pq} = \frac{\sigma_{a-1,d }\sigma_{b,c-1}}{\sigma_{b,d} \sigma_{a-1,c-1}}.
\end{equation}{}
Using this identity we can write the map between world-sheet associahedron and ABHY associahedron as follows,
\begin{equation}\label{worldABHYmap}
     X_{ij} = \sum_{\substack{ i\leq k \leq j-2\\ j \leq \ell \leq i-2 }} c_{k\ell} \frac{\sigma_{i-1,\ell }\sigma_{k,j-1}}{\sigma_{k,\ell} \sigma_{i-1,j-1}} = \sum_{\substack{ i\leq k \leq j-2\\ j \leq \ell \leq i-2 }} c_{k\ell} \prod_{\substack{i \leq  p \leq k \\ j\leq q \leq \ell} } u_{pq}\,.
\end{equation}

It can be checked that this map reduces to the map given in \cite{Arkani-Hamed:2017mur} for reference triangulation $T=\{(13),(14),\ldots,(1,n-1)\}$. We have verified that for $n=6$ with reference triangulations $\{(13),(14),(46)\}$ and $\{(13),(35),(15)\}$ the scattering equations give the above map (\ref{worldABHYmap}).

\section{Five Point Push Forward} 
\label{mrunmay1}
We begin with the following form : 
\begin{equation}\label{form}
\begin{split}
    \widehat{\Omega} &= (\mathrm{d}\log u_{13}+\mathrm{d}\log f_{13})\wedge (\mathrm{d}\log u_{14}+\mathrm{d}\log f_{14}) + (\mathrm{d}\log u_{35}+\mathrm{d}\log f_{35})\wedge (\mathrm{d}\log u_{13}+\mathrm{d}\log f_{13})\\ & + (\mathrm{d}\log u_{14}+\mathrm{d}\log f_{14})\wedge (\mathrm{d}\log u_{24}+\mathrm{d}\log f_{24})+ (\mathrm{d}\log u_{25}+\mathrm{d}\log f_{25})\wedge (\mathrm{d}\log u_{35}+\mathrm{d}\log f_{35}) \\ & + (\mathrm{d}\log u_{24}+\mathrm{d}\log f_{24})\wedge (\mathrm{d}\log u_{25}+\mathrm{d}\log f_{25}),
\end{split}
\end{equation}
where the $f_{ij}$ are given by 
\begin{align}\label{fs}
    f_{13} &= c_{13}+ c_{14}u_{14} \cr
    f_{14} &= c_{14} + c_{24}u_{24}  \cr
    f_{24} &= c_{13}u_{25} + c_{14}u_{14}u_{25} + c_{24}u_{14}\cr
     f_{25} &= c_{13}u_{24} + c_{14}\cr
      f_{35} &= c_{14}u_{25} + c_{24}\,.
\end{align}
After some tedious algebra we simplify the form in (\ref{form}) to get,

\begin{small}
\begin{equation}
   \begin{split}
        \widehat{\Omega} &=   \frac{\mathrm{d}u_{13}\wedge \mathrm{d}u_{14} }{u_{13}(1-u_{13})u_{14}(1-u_{14})} \\ &\hspace{0.5cm} +  \frac{c_{13}c_{14}c_{24}\left(\left(c_{14}+\left(c_{24}-u_{13}f_{13}\right)\right)\left(c_{13}- \left(c_{24}-u_{14}f_{14}\right)\right)\right)}{f_{13}f_{14}f_{24}f_{25}f_{35}(1-u_{13}u_{14})}  \mathrm{d}u_{13}\wedge \mathrm{d}u_{14}
   \end{split}
\end{equation}
\end{small}

%
%

In our choice of gauge $\sigma_{1}=0,\sigma_{4}=1,\sigma_{5}=\infty$,
\begin{align}
    u_{13} &= \frac{\sigma_{2}}{\sigma_{3}}, & u_{14} &= \sigma_{3}, & u_{25} &= 1- \sigma_{2},  & u_{35} &= \frac{1-\sigma_{3}}{1-\sigma_{2}}, & u_{24}&= \frac{\sigma_{3}-\sigma_{2}}{(1-\sigma_{2})\sigma_{3}}\,,
\end{align}
and 
\begin{align}\label{fsig}
    f_{13} &= c_{13}+ c_{14}\sigma_{3} \cr
    f_{14} &= c_{14} + c_{24}\frac{\sigma_{3}-\sigma_{2}}{(1-\sigma_{2})\sigma_{3}}  \cr
    f_{24} &= c_{13}(1- \sigma_{2}) + (c_{14}(1- \sigma_{2}) + c_{24})\sigma_{3}\cr
     f_{25} &= c_{13}\frac{\sigma_{3}-\sigma_{2}}{(1-\sigma_{2})\sigma_{3}} + c_{14}   \cr
     f_{35} &= c_{14}(1- \sigma_{2}) + c_{24}\,.
\end{align}
%
%
After some algebra, the form (\ref{form}) can be re-expressed as : 
\begin{equation}
\label{Omegahatrexpress}
    \widehat{\Omega} =  \left[ \frac{1}{\sigma_{2}(\sigma_{2}-\sigma_{3})(\sigma_{3}-1)} + R \right] \mathrm{d}\sigma_{2}\wedge \mathrm{d}\sigma_{3}.
\end{equation}{}
Where 
\begin{align}
    R = \frac{c_{13}c_{14}c_{24}\left(\left(c_{14}+\left(c_{24}-u_{13}f_{13}\right)\right)\left(c_{13}- \left(c_{24}-u_{14}f_{14}\right)\right)\right)}{f_{13}f_{14}f_{24}f_{25}f_{35}(1-u_{13}u_{14})u_{14}}
\end{align}{}
%
%
Notice that for our first choice of kinematics where $c_{14}=0$, we have $N_1=0$ and we see from \eqref{Omegahatrexpress} that the form $\widehat\Omega$ reduces to the Parke-Taylor form and hence matches $\omega_{5}^{\textrm{ws}}$ as was claimed in Section \ref{scrutiny}.

In the special choice of kinematics $c_{13}=c_{14} = c_{24}= 1$ the form is given by 
\begin{equation}\label{spform}
    \widehat{ \Omega} = \left(\frac{1}{\sigma_{2}(\sigma_{2}-\sigma_{3})(\sigma_{3}-1)}+\frac{1}{(2-\sigma_{2}) (1+ \sigma_{3}) (1-\sigma_{2} + (2- \sigma_{2} )\sigma_{3} ) } \right) \mathrm{d}\sigma_{2}\wedge \mathrm{d}\sigma_{3}.
\end{equation}
%
%
For a general choice of kinematics the solutions to the scattering equations \eqref{scattereq} are :
\begin{tiny}
\begin{align}
    \sigma_{2} &= {\frac{\sqrt{(c_{13} (c_{14}+c_{24})+X_{14} (c_{14}+X_{13})-c_{24} X_{13})^2-4 c_{14} X_{13}X_{14}
   (c_{13}-c_{24}+X_{14})}+c_{13} (c_{14}+c_{24})+c_{14}X_{14}-c_{24} X_{13}+X_{13}X_{14}}{2 c_{14}
   (c_{13}-c_{24}+X_{14})}}\nonumber\\
   \sigma_{3} &=  {-\frac{\sqrt{(c_{13} (c_{14}+c_{24})+X_{14} (c_{14}+X_{13})-c_{24} X_{13})^2-4 c_{14} X_{13}X_{14}
   (c_{13}-c_{24}+X_{14})}+c_{13} (c_{14}+c_{24})-c_{14}X_{14}-c_{24} X_{13}+X_{13}X_{14}}{2 c_{14}
   (c_{14}+c_{24}-X_{13})}}
\end{align}
\end{tiny}
and
\begin{tiny}
\begin{align}
    \sigma_{2} &= {\frac{-\sqrt{(c_{13} (c_{14}+c_{24})+X_{14} (c_{14}+X_{13})-c_{24} X_{13})^2-4 c_{14} X_{13}X_{14}
   (c_{13}-c_{24}+X_{14})}+c_{13} (c_{14}+c_{24})+c_{14}X_{14}-c_{24} X_{13}+X_{13}X_{14}}{2 c_{14}
   (c_{13}-c_{24}+X_{14})}} \nonumber\\
     \sigma_{3} &= {\frac{\sqrt{(c_{13} (c_{14}+c_{24})+X_{14} (c_{14}+X_{13})-c_{24} X_{13})^2-4 c_{14} X_{13}X_{14}
   (c_{13}-c_{24}+X_{14})}-c_{13} (c_{14}+c_{24})+c_{14}X_{14}+c_{24} X_{13}-X_{13}X_{14}}{2 c_{14}
   (c_{14}+c_{24}-X_{13})}}
\end{align}
\end{tiny}
For the special choice of kinematics $c_{13}=c_{14} = c_{24}= 1$  the solutions become,
\begin{align}\label{spsol1}
    \sigma_{2} &= \frac{\sqrt{(X_{13} (X_{14}-1)+X_{14}+2)^2-4 X_{13} X_{14}^2}+X_{13} (X_{14}-1)+X_{14}+2}{2
   X_{14}}\nonumber\\
   \sigma_{3} &= \frac{\sqrt{(X_{13} (X_{14}-1)+X_{14}+2)^2-4 X_{13} X_{14}^2}+X_{13} (X_{14}-1)-X_{14}+2}{2
   (X_{13}-2)}
\end{align}
and 
\begin{align}\label{spsol2}
    \sigma_{2} &= \frac{-\sqrt{(X_{13} (X_{14}-1)+X_{14}+2)^2-4 X_{13} X_{14}^2}+X_{13} (X_{14}-1)+X_{14}+2}{2
   X_{14}}\nonumber \\
     \sigma_{3} &= \frac{\sqrt{(X_{13} (X_{14}-1)+X_{14}+2)^2-4 X_{13} X_{14}^2}+X_{13}
   (-X_{14})+X_{13}+X_{14}-2}{2 (2-X_{13})}
\end{align}
The push forward of (\ref{spform}) along each of above solutions (\ref{spsol1},\ref{spsol2}) gives 
\begin{equation}\label{ampsp}
 \left(    \frac{-X_{13} X_{14}+2 (2-X_{13})+4 X_{14}}{(2-X_{13}) X_{13} (2-X_{14}) X_{14}
   (-X_{13}+X_{14}+1)}  \right)\mathrm{d}X_{13}\wedge\mathrm{d}X_{14}.
\end{equation}{}
which is equal to the following for the special choice of kinematics, $c_{13}=c_{14} = c_{24}= 1$ 
\begin{equation}\label{amp}
   \left( \frac{1}{X_{13} X_{14}}+\frac{1}{X_{13} X_{35}}+\frac{1}{X_{14} X_{24}}+\frac{1}{X_{24}
   X_{25}}+\frac{1}{X_{25} X_{35}} \right) \mathrm{d}X_{13}\wedge\mathrm{d}X_{14}.
\end{equation}
\section{Linear diffeomorphism for hyper-cube case}\label{linear}
As we saw in Section \ref{pardiss} there is a convex realisation of ${\cal S}^{\{14,5n,\ \dots,\ (\frac{n}{2}+1)(\frac{n}{2}+4)\}}_n$ as a facet of ${\cal A}_{n}$ given by the constraints \eqref{145nembeddingequations} where the $d_k$ take the values given in Claim \ref{claim145nembedding}. 
%
The corresponding polytope coincides with a co-dimension $\frac{n}{2}-1$ facet of ${\cal A}_{n-3}$ given by setting some of the $X_{ij}$ to zero as stated in Claim \ref{claim145nembedding}.
%
This boundary is diffeomorphic to the boundary of $\widetilde{\cal A}_{n-3}$ obtained by setting
a set of Pl\"ucker co-ordinates to zero :
\begin{equation}\label{vanipluk}
u_{13}=0,\quad \&\quad u_{4+k,n-k}=0\quad\text{for}\quad k\in[0,1,\ldots,\frac{n}{2}-3]\,.
\end{equation}
\begin{claim}
\label{lineardiffeoclaim}
Scattering equations provide a linear diffeomorphism between the aforementioned facet of $\widetilde{\cal A}_{n-3}$ and the convex realisation of ${\cal S}_{n}^{Q}$
where $Q=\{14,5n,\ \dots,\ (\frac{n}{2}+1)(\frac{n}{2}+4)\}$ : 
\begin{align}
X_{ij}={\cal D}_{ij}\ u_{ij},\,\,\, \forall\ (ij)\ &\in\ \{14, (5+k,n-k),\ 3n,\  (4+k,n-k-1)\}\,,
\end{align}
where $k=\{0,\ldots,\frac{n}{2}-4\}$ and ${\cal D}_{ij}={\cal D}_{ij}(\{c_{mn}\})$ are  linear sums of $c_{ij}$'s as given in \eqref{linxu}.
%
\end{claim}
\begin{proof}
Our proof will go in two stages. In the first stage we will use scattering equations to write planar variables labelled by elements of $Q$ in terms of world-sheet co-ordinates and show that due to \eqref{vanipluk} such a map is linear between $X_{ij}$ and the corresponding $u_{ij}$. In the second stage we will show that the map between $u_{ij}$ and corresponding $X_{ij}$ is linear also for $(ij)\in\{3n,4 n-1,\ldots,\frac{n}{2} \frac{n}{2}+3\}$.
\\
\\
{\textit{Stage 1}} : Let us first consider $X_{14}$. The scattering equations give :
\begin{align}
\label{genx14}
X_{14}=\sum_{1\le i\le 2}\,\sum_{4\le j\le n-1}\,\frac{\sigma_{i,3}}{\sigma_{i,j}}c_{ij}\,.
\end{align}
We express the world-sheet coordinates $\sigma_i$ that appear on the RHS of \eqref{genx14} in terms of Pl\"ucker coordinates $u_{1j}$ and then substitute $u_{1j}=1$ for $j=5,\ldots,n-1$ which arise from the respective crossing relations and impose $u_{14}=0$ from \eqref{vanipluk} to obtain the linear map between $X_{14}$ and $u_{14}$ in \eqref{linxu}. Since for $n=6$, $Q$ is just $\{14\}$ for which we obtained the linear map for general $n$ let us now consider $n>6$. 
\\
\\
Let us now consider $X_{5n}$. From the scattering equations we have :
\begin{align}
\label{genX5n1st}
X_{5n}=\sum_{1\le i\le 4}\,\sum_{6\le j\le n-1}\,\frac{\sigma_{5,j}}{\sigma_{i,j}}c_{ij}\,.
\end{align}
We will make use of the identity :
\begin{align}
\label{id1}
\frac{\sigma_{5,j}}{\sigma_{i,j}}=\prod_{m=i+1}^5\prod_{k=1}^{n-j}u_{m,j+k}\,,\quad 1\le i\le4,\quad 6\le j\le n-1\,.
\end{align}
Using  $u_{4,n}=0$ and  $u_{5,n-1}=0$ from \eqref{vanipluk}, it is easy to see that the RHS is non-vanishing only when $i=4$ and $j=n-1$. This gives the linear map between $X_{5n}$ and $u_{5n}$ in \eqref{linxu}.
\\
\\
 %
 Now that we have shown the linear map for $X_{ij}$ corresponding to the labels $(14),(5n) \in Q$, let us consider the scattering equations for $X_{ab}$ with $(a b)\in Q$ and $a=5+k,\ b=n-k$ such that $k\in \{1,\ldots,\frac{n}{2}-4\}$ :
\begin{align}\label{scateqpar}
X_{ab}\ &=\sum_{\substack{1\leq i\leq a-1 \\ \\ a+1\leq j\leq\ b-1}}\frac{\sigma_{aj}}{\sigma_{ij}}\ c_{ij}\ +\ 
\sum_{\substack{a\leq i\leq b-1 \\ \\ b\leq j\leq n-1}}\frac{\sigma_{ib-1}}{\sigma_{ij}}c_{ij} + \sum_{\substack{1\leq i\leq a-1 \\ \\ b\leq j\leq n-1}}\frac{\sigma_{ab-1}}{\sigma_{ij}}c_{ij}\,.
\end{align}
To analyse the first summand we find the following identity useful :
\begin{align}
\frac{\sigma_{aj}}{\sigma_{ij}}=\prod_{\ell=i+1}^a\prod_{m=j+1}^{n}u_{\ell m},\quad\forall\ 1\ \leq\ i\ \leq a-1,\quad a+1\le j\le b-1\,.
\end{align}
Using $u_{a(b-1)}=0$ and $u_{(a-1)b}=0$ from \eqref{vanipluk} we see that the only non-vanishing contribution to the RHS comes from $i=a-1,j=b-1$. Using $u_{am}=1$ $\forall\ m\ \in\ \{b+1,\ \dots,\ n\}$ we see that the contribution of this term to the first summand is $u_{ab}\ c_{(a-1)(b-1)}$.
\\
\\
To analyse the second summand \eqref{scateqpar} the following identity is useful :
\begin{align}
\label{id2ndterm}
\frac{\sigma_{i,b-1}}{\sigma_{ij}}=\prod_{c=1}^i\,\prod_{d=b}^j\,u_{cd},\quad a\le i\le b-1,\quad b\le j\le n-1\,.
\end{align}
It can be seen that $\forall i,j$ in the above range, there is at least one $u_{cd}$ on the RHS such that $c+d=4+n$ and $4\le c\le\frac{n}{2}+1$ which we know from \eqref{vanipluk} is zero. Thus the contribution from the second summand is zero.
\\
\\
We now come to the third summand in \eqref{scateqpar}. For $1\le i\le a-1, b< j \le n-1$ :
\begin{align}
\label{sum3term}
\frac{\sigma_{a,b-1}}{\sigma_{i,j}}=\prod_{c\,=1}^a\,u_{c\,b}\prod_{d\,=\,b+1}^j\,u_{1d}\prod_{r=i+1}^a\prod_{s=b+1}^n\,u_{rs}\prod_{r=2}^i\prod_{s=b+1}^j\,u_{rs}\,.
\end{align}
The first product on the RHS contains $u_{a-1,b}$ which is zero \eqref{vanipluk}. For $1\le i\le a-1, j=b$ we make use of the following identity :
\begin{align}
\label{sum3term2ndcase}
\frac{\sigma_{a,b-1}}{\sigma_{i,b}}=\prod_{c\,=1}^a\,u_{c\,b}\,\prod_{\ell=i+1}^a\prod_{m=b+1}^{n}u_{l,m},\quad 1\le i\le a-1\,.
\end{align}
This is also zero since the RHS contains  $u_{a-1,b}$.
Thus the contribution of the third summand in \eqref{scateqpar} is also zero. This completes Stage 1 of our proof where we have shown that $\forall\ (a,b)\in Q=\{14,5n,\ \dots,\ (\frac{n}{2}+1)(\frac{n}{2}+4)\}$ there is a linear map between $X_{ab}$ and $u_{ab}$ as given in the first two lines of \eqref{linxu}. 

We now come to the second stage of the proof where we show that a similar linear map exists also for labels $(ij)\in\{3n,4 n-1,\ldots,\frac{n}{2} \frac{n}{2}+3\}$.
\\
\\
{\textit{Stage 2}} : 
Using the first constraint in \eqref{momentumconsassociahedron145m}, the linear map between $X_{14}$ and $u_{14}$ in \eqref{linxu}, and the crossing relation $u_{14}=1-u_{3n}$
we obtain 
\begin{align}
X_{3n}=u_{3n}\sum_{4\le j\le n-1}\,(c_{1j}+c_{2j})\,.
\end{align}
as given in the third line of \eqref{linxu}.
%
We note the crossing relation :
\begin{equation}\label{pluksim}
u_{3+k,n-k}=1-u_{3+k+1,n-k+1}\quad \forall\ k\ \in\ [1,\dots,\frac{n-6}{2}]\,,
\end{equation}
where we have used the cyclicity of polygon vertices $n + 1\ =:\ 1$. This is rather straightforward as $\forall\ k\ \in\ [1,\dots,\frac{n-6}{2}]$ every dissection except $(3+k,n-k)$ intersects at least one of the $(i,j)$ in the set $\{(1,3),(4,n),\ldots,(\frac{n}{2}+1,\frac{n}{2}+3)\}$ for which the corresponding $u_{ij}=0$ from \eqref{vanipluk}. Using constraint equations given in the second line of \eqref{momentumconsassociahedron145m}, linear maps in the second line of \eqref{linxu}, and the crossing relations in \eqref{pluksim} it follows that the map between $u_{ij}$ and corresponding $X_{ij}$'s is linear for $\{X_{4 n-1},\dots,X_{\frac{n}{2} \frac{n}{2}+3}\}$ :
%
\begin{align}
X_{ij}\ =\ c_{ij}\ u_{ij}\,,\,\,\,\forall\quad i=4+k, j=n-k-1,\,\,\,\text{for}\,\,k=\{0,1,\ldots \frac{n}{2}-4\}\,,
\end{align}
as stated in the last line of \eqref{linxu}. This completes the proof of Claim \ref{lineardiffeoclaim}.
\end{proof}

%

%
\vskip 1cm


\begin{thebibliography}{99}

\bibitem{Banerjee:2018tun}
  P.~Banerjee, A.~Laddha and P.~Raman,
  ``Stokes polytopes: the positive geometry for $\phi^{4}$ interactions,''
  JHEP {\bf 1908} (2019) 067,
  arXiv:1811.05904 [hep-th]\,.
  
\bibitem{Arkani-Hamed:2017mur}
  N.~Arkani-Hamed, Y.~Bai, S.~He and G.~Yan,
  ``Scattering Forms and the Positive Geometry of Kinematics, Color and the world-sheet,''
  JHEP {\bf 1805} (2018) 096,
  arXiv:1711.09102 [hep-th]\,.
  
  \bibitem{1807ppp}
  Y.~Palu, V.~Pilaud and P-G.~Plamondon, 
  ``Non-kissing and non-crossing complexes for locally gentle algebras," 
  arXiv:1807.04730\,.
  
 \bibitem{1906ppp}
  A.~Padrol, Y.~Palu, V.~Pilaud and P-G.~Plamondon, 
  ``Associahedra for finite cluster type algebra and minimal relations between g-vectors,''  
  arXiv:1906.06861 [math.RT]\,.
  
  \bibitem{HughThomas}
   V.~Bazier-Matte, G.~Douville, K.~Mousavand, H.~Thomas and E.~Yildirim, 
   ``ABHY Associahedra and Newton polytopes of F-polynomials for finite type cluster algebras," 
  arXiv:1808.09986 [math.RT]\,.
  
\bibitem{Arkani-Hamed:2013jha} 
  N.~Arkani-Hamed and J.~Trnka,
  ``The Amplituhedron,''
  JHEP {\bf 1410}, 030 (2014)
  arXiv:1312.2007 [hep-th]\,.
  
\bibitem{Arkani-Hamed:2013kca} 
  N.~Arkani-Hamed and J.~Trnka,
  ``Into the Amplituhedron,''
  JHEP {\bf 1412}, 182 (2014)
  arXiv:1312.7878 [hep-th]\,.
  
\bibitem{Franco:2014csa} 
  S.~Franco, D.~Galloni, A.~Mariotti and J.~Trnka,
  ``Anatomy of the Amplituhedron,''
  JHEP {\bf 1503}, 128 (2015)
  arXiv:1408.3410 [hep-th]\,.   
  
\bibitem{Arkani-Hamed:2018rsk} 
  N.~Arkani-Hamed, C.~Langer, A.~Yelleshpur Srikant and J.~Trnka,
  ``Deep Into the Amplituhedron: Amplitude Singularities at All Loops and Legs,''
  Phys.\ Rev.\ Lett.\  {\bf 122}, no. 5, 051601 (2019)
  arXiv:1810.08208 [hep-th]\,.
  
\bibitem{Arkani-Hamed:2017vfh} 
  N.~Arkani-Hamed, H.~Thomas and J.~Trnka,
  ``Unwinding the Amplituhedron in Binary,''
  JHEP {\bf 1801}, 016 (2018)
  arXiv:1704.05069 [hep-th]\,.


  
\bibitem{He:2018pue}
  S.~He, G.~Yan, C.~Zhang and Y.~Zhang,
  ``Scattering Forms, world-sheet Forms and Amplitudes from Subspaces,''
  JHEP {\bf 1808} (2018) 040,
  arXiv:1803.11302 [hep-th]\,.
 
\bibitem{delaCruz:2017zqr} 
  L.~de la Cruz, A.~Kniss and S.~Weinzierl,
  ``Properties of scattering forms and their relation to associahedra,''
  JHEP {\bf 1803}, 064 (2018)
  arXiv:1711.07942 [hep-th]\,.
  

  
\bibitem{Raman:2019utu}
  P.~Raman,
  ``The positive geometry for $\phi^{p}$ interactions,''
  JHEP {\bf 1910}, 271 (2019)
  arXiv:1906.02985 [hep-th]\,.
  
  
  
\bibitem{Jagadale:2019byr}
  M.~Jagadale, N.~Kalyanapuram and A.~P.~Balakrishnan,
  ``Accordiohedra as Positive Geometries for Generic Scalar Field Theories,''
  arXiv:1906.12148 [hep-th]\,.
  
   
\bibitem{Cachazo:2014xea} 
  F.~Cachazo, S.~He and E.~Y.~Yuan,
  ``Scattering Equations and Matrices: From Einstein To Yang-Mills, DBI and NLSM,''
  JHEP {\bf 1507}, 149 (2015)
  arXiv:1412.3479 [hep-th]\,.
  
\bibitem{Baadsgaard:2015ifa} 
  C.~Baadsgaard, N.~E.~J.~Bjerrum-Bohr, J.~L.~Bourjaily and P.~H.~Damgaard,
  ``Scattering Equations and Feynman Diagrams,''
  JHEP {\bf 1509}, 136 (2015)
  arXiv:1507.00997 [hep-th]\,.
  
\bibitem{Baadsgaard:2016fel} 
  C.~Baadsgaard, N.~E.~J.~Bjerrum-Bohr, J.~L.~Bourjaily and P.~H.~Damgaard,
  ``String-Like Dual Models for Scalar Theories,''
  JHEP {\bf 1612}, 019 (2016)
  arXiv:1610.04228 [hep-th]\,.

  \bibitem{accoref}
T.~Manneville and V.~Pilaud,
  ``Geometric realizations of the accordion complex of a dissection,"
  Discrete \& Computational Geometry, 61(3) 507-540, 2019,
  arXiv:1703.09953\,.
  
  
  
  \bibitem{Stasheff1}
J. ~Stasheff,
 ``Homotopy Associativity of H-Spaces. I.'' (1963)
  	``Transactions of the American Mathematical Society,"
  	108, no. 2, 1963, 275-292\,.
  
  
\bibitem{Stasheff2}
J.~Stasheff,
  ``Homotopy Associativity of H-Spaces. II.'' (1963)
  	``Transactions of the American Mathematical Society,"
  	108, no. 2, 1963, 293-312\,. 
  
  
  \bibitem{Baryshnikov}
  Y.~Baryshnikov,
 ``On Stokes sets," 
  	New developments in singularity theory, 21, 2001, 65-86\,.
  
    \bibitem{Chapoton}
   F.~Chapoton,
  ``Stokes posets and serpent nest,"
   Discrete Mathematics \& Theoretical Computer Science, 18(3), 2016,
  arXiv:1505.05990 [math.RT]\,.

 \bibitem{devadoss}
  S.~L.~Devadoss, T.~Heath and W.~Vipismakul, 
  ``Deformations of bordered surfaces and convex polytopes,"
  Notices of the AMS 58 (2011) 530-541, 
  arXiv:1002.1676\,.
  
\bibitem{Kalyanapuram:2019nnf}
N.~Kalyanapuram,
  ``Stokes Polytopes and Intersection Theory,''
  arXiv:1910.12195 [hep-th]\,.
  
  
  
  \bibitem{cohomologyannals}
  P.~Etingof, A.~Henriques, J.~Kamnitzer, E.~M.~Rains, 
  ``The cohomology ring of the real locus of the moduli space of stable curves of genus 0 with marked points," 
  Ann. of Math (2) 171:2, 731-777 (2010)\,.
  
  \bibitem{Devadoss}
  S.~Devadoss, 
  ``Tessellations of Moduli Spaces and the Mosaic Operad," 
  Contemp. Math. 239 (1999), 91-114
  arXiv:9807010 [math.AG]\,.
  
 
  
  
  
\bibitem{Mizera:2019gea} 
  S.~Mizera,
  ``Aspects of Scattering Amplitudes and Moduli Space Localization,''
  arXiv:1906.02099 [hep-th]\,.
  
 
\bibitem{Mizera:2017rqa} 
  S.~Mizera,
  ``Scattering Amplitudes from Intersection Theory,''
  Phys.\ Rev.\ Lett.\  {\bf 120}, no. 14, 141602 (2018)
  arXiv:1711.00469 [hep-th]\,.


\bibitem{Cachazo:2013iea} 
  F.~Cachazo, S.~He and E.~Y.~Yuan,
  ``Scattering of Massless Particles: Scalars, Gluons and Gravitons,''
  JHEP {\bf 1407}, 033 (2014)
  arXiv:1309.0885 [hep-th]\,.
  
\end{thebibliography}
\end{document}